\documentclass[12pt,conference]{IEEEtran}
\onecolumn
\usepackage{amsmath, amssymb,amsthm,cite}
\usepackage[dvips]{graphics}
\usepackage[utf8x]{inputenc}
\usepackage{graphicx}
\usepackage{psfrag}
\usepackage{bbm}
\usepackage{subfigure}
\usepackage{tabularx}
\usepackage{dblfloatfix}
\usepackage{mathtools}
\usepackage{algpseudocode}
\usepackage{algorithm}
\usepackage{hyperref}
\usepackage{color}
\usepackage{tcolorbox}
\usepackage{etoolbox}
\usepackage{mleftright}
\usepackage{algorithm}
\usepackage{algpseudocode}
 
\DeclareMathOperator*{\nn}{\nonumber}

\allowdisplaybreaks
\newcommand{\RNum}[1]{\uppercase\expandafter{\romannumeral #1\relax}}
\newtheorem{lemma}{Lemma}
\newtheorem{theorem}{Theorem}

\theoremstyle{definition}
\newtheorem{remark}{Remark}

\makeatletter
\def\blfootnote{\gdef\@thefnmark{}\@footnotetext}
\makeatother
\def\cX{{\mathcal X}}
\def\cY{{\mathcal Y}}
\def\cS{{\mathcal S}}

\def\cQ{{\mathcal Q}}

\newcommand{\pr}[1]{\left(#1\right)}
\DeclarePairedDelimiterX{\infdivx}[2]{(}{)}{%
  #1\;\delimsize\|\;#2%
}
\newcommand{\infdiv}{D\infdivx}

\title{The Duality Upper Bound for Finite-State Channels with Feedback}

\author{
\IEEEauthorblockN{Bashar Huleihel} \IEEEauthorblockA{
basharh@post.bgu.ac.il }
\and
\IEEEauthorblockN{Oron Sabag}
\IEEEauthorblockA{
oron.sabag@mail.huji.ac.il}
\and
\IEEEauthorblockN{Ziv Aharoni}
\IEEEauthorblockA{
zivah@post.bgu.ac.il}
\and
\IEEEauthorblockN{Haim H. Permuter} \IEEEauthorblockA{
haimp@bgu.ac.il}
}

\begin{document}
\maketitle
\begin{abstract}
This paper investigates the capacity of finite-state channels (FSCs) with feedback. We derive an upper bound on the feedback capacity of FSCs by extending the duality upper bound method from mutual information to the case of directed information. The upper bound is expressed as a multi-letter expression that depends on a test distribution on the sequence of channel outputs. For any FSC, we show that if the test distribution is structured on a $Q$-graph, the upper bound can be formulated as a Markov decision process (MDP) whose state being a belief on the channel state. In the case of FSCs and states that are either unifilar or have a finite memory, the MDP state simplifies to take values in a finite set. Consequently, the MDP consists of a finite number of states, actions, and disturbances. This finite nature of the MDP is of significant importance, as it ensures that dynamic programming algorithms can solve the associated Bellman equation to establish analytical upper bounds, even for channels with large alphabets. We demonstrate the simplicity of computing bounds by establishing the capacity of a broad family of Noisy Output is the State (NOST) channels as a simple closed-form analytical expression. Furthermore, we introduce novel, nearly optimal analytical upper bounds on the capacity of the Noisy Ising channel.
\end{abstract}

\blfootnote{This paper was presented in part at the International Zurich Seminar on Information and Communication (IZS) \cite{Sabag_DB_FB}.}

\section{Introduction}\label{sec:intro}
Finite-state channels (FSCs) \cite{McMillan1953TheBT,Shannon_FSC,Blackwell58} are commonly used to model scenarios in which the channel or the system have memory, as encountered in many practical applications such as wireless communication \cite{FSC_Wirless1,FSC_Wirless2,Wang95_FSC_usful_for_radiochannels}, magnetic recording \cite{FSC_Magnetic}, and molecular communication \cite{MolecFSCTransComm,MolecularSurvey}. Despite their significance in theory and practice, both the capacity and the feedback capacity of FSCs are characterized by multi-letter formulas that are hard to compute. This paper investigates the capacity of FSCs with feedback (see Fig. \ref{fig:FSC}), and develops a methodology to derive computable upper bounds on their feedback capacity.

A common approach for computing the feedback capacity of FSCs is via the Markov decision process (MDP) formulation of the capacity expression \cite{Permuter06_trapdoor_submit,Yang05,TatikondaMitter_IT09}. In certain cases, employing dynamic programming (DP) methods or reinforcement learning (RL) algorithms yields a numerical solution that can be translated into a conjectured optimal solution. To conclude its optimality and derive a corresponding analytical capacity expression, a solution for the involved Bellman equation is required, as has been done for several particular examples \cite{Chen05,PermuterCuffVanRoyWeissman08,Ising_channel,Sabag_BEC,trapdoor_generalized,Ising_artyom_IT,Sabag_BIBO_IT}. However, due to the inherent high complexity associated with the continuous states and actions of the involved MDP, obtaining such an analytical solution is generally infeasible for the majority of channels, particularly when the channel alphabets extend beyond binary.

In a recent paper \cite{Sabag_UB_IT}, the authors introduced an upper bound on the feedback capacity of unifilar FSCs. The upper bound is expressed as a single-letter formula and holds for any choice of a \textit{$Q$-graph}\footnote{The $Q$-graph, introduced in \cite{Sabag_UB_IT}, is an auxiliary directed graph that is used to map output sequences onto one of the auxiliary graph nodes (For instance, see Fig. \ref{fig:1Markov}).}.
Furthermore, it was demonstrated that the bound can be transformed into a conventional convex optimization problem \cite{OronBasharfeedback}. Although this bound has led to new capacity results, its analytical computation remains challenging due to the requisite verification of the Karush–Kuhn–Tucker (KKT) conditions. This complexity becomes particularly pronounced when channel parameters involve large alphabets. The recent development of RL algorithms for (numerical) computation and conjecturing of optimal solutions for feedback capacity in scenarios involving large alphabets \cite{aharoni2022feedback} has motivated the current paper. The approach proposed in this paper is applicable to any FSC and offers an alternative, significantly simpler method for deriving analytical upper bounds. 
 
Motivated by the above challenges, the current paper focuses on simple derivation of analytical upper bounds that are suitable to channels with large alphabets. The derivation is based on the dual upper bounding technique \cite{Andrew_ISI,MIMO_dual}, adapted to directed information. The resulting duality bound is a multi-letter formula that is a function of a test distribution on the channel outputs that needs to be optimized. Our main contribution is that if the test distribution is structured on a Q-graph, the upper bound can be formulated as an MDP. The MDP formulation holds for any FSC and a test distribution that is structured on a $Q$-graph. Further, for both unifilar FSCs and finite-memory state channels, the MDP formulation consists of a finite number of states, actions, and disturbances. As a result, simple numerical and analytical MDP tools can be applied to solve the associated MDP, owing to the finite alphabets in the formulation.

\begin{figure}[t]
\centering
    \psfrag{E}[][][.95]{Encoder}
    \psfrag{D}[][][.95]{Decoder}
    \psfrag{C}[][][0.9]{$P_{S^+,Y|X,S}(s_t,y_t|x_t,s_{t-1})$}
    \psfrag{V}[][][.78]{Unit-Delay}
    \psfrag{M}[][][1]{$m$}
    \psfrag{Y}[][][1]{$y_t$}
    \psfrag{O}[][][1]{$\hat{m}$}
    \psfrag{Z}[][][1]{$y_{t-1}$}
    \psfrag{X}[][][1]{$x_t$}
    \includegraphics[scale = 0.7]{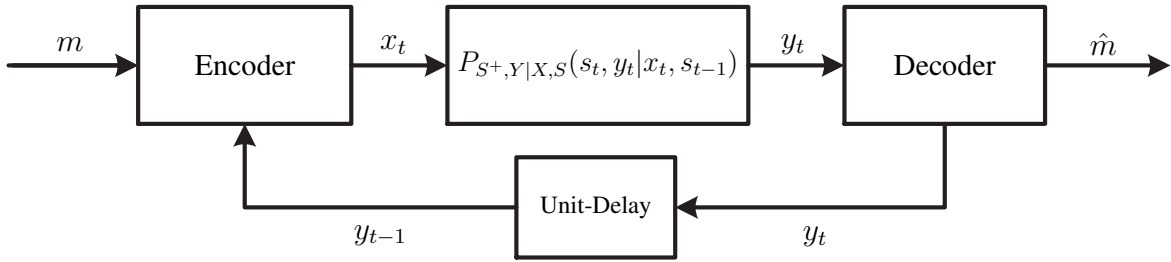}
    \caption{Finite-state channel with feedback.}
    \label{fig:FSC}
\end{figure}
The choice of the test distribution is crucial as it directly affects the performance of the upper bound. As in the case of a memoryless channel, if the test distribution matches the optimal output distribution, the resulting upper bound is tight. In most cases, the optimal output distribution for channels with memory is not i.i.d.. While Markov test distributions are commonly chosen \cite{Dual_Andrew_J,Duality_Kremer,Vontobel_Dual}, it has been demonstrated that certain channels do not admit a Markovian structure of any finite order. To address this limitation, we adopt \textit{graph-based test distributions} introduced in \cite{Huleihel_Sabag_DB}. These test distributions are structured on $Q$-graphs and provide a generalization of the standard Markov test distributions, allowing for a variable-order Markov process on the channel output sequence.

To illustrate the process of deriving analytical upper bounds using our proposed methodology, we study the feedback capacity of several FSCs, including the Noisy Output is the STate (NOST) channels \cite{shemuel2022feedback} and the Noisy Ising (N-Ising) channels, a generalized version of the well-known Ising channel \cite{Berger90IsingChannel}. In particular, we analyze a broad family of NOST channels, establishing their capacity as a closed-form analytical expression. Furthermore, we establish new upper bounds on the capacity of the N-Ising channel and demonstrate their near-tightness. For all these examples, we derive the upper bounds by selecting a specific graph-based test distribution and subsequently solving the associated MDP problem.

The remainder of the paper is organized as follows. Section \ref{sec:prelimi} provides our notations and the model definition. Section \ref{sec: preliminaries} introduces the dual capacity upper bound and provides background on $Q$-graphs. Section \ref{sec: main_results} outlines our main results. Section \ref{sec: DP_Main} provides a concise overview of infinite-horizon DP and introduces our MDP formulation of the dual capacity upper bound for FSCs with feedback. Section \ref{sec: Q_graph_exploration} presents several techniques to explore for $Q$-graphs that yield good performance upper bounds. Section \ref{sec:examples} presents our analytic results on the capacity of several FSCs. Finally, our conclusion appears in Section \ref{sec:conclusion}. To maintain the flow of the presentation, some proofs are given in the appendices.

\section{Notation and Model Definition}\label{sec:prelimi}
In this section, we introduce our notation and define the FSC model.

\subsection{Notation}
Throughout this paper, we use the following notations. The set of natural numbers, excluding zero, is denoted by $\mathbb{N}$, and $\mathbb{R}$ denotes the set of real numbers. Random variables will be denoted by capital letters, and their realizations will be denoted by lower-case letters, e.g., $X$ and $x$, respectively. Calligraphic letters denote sets, e.g.,  $\mathcal{X}$. We use the notation $X^n$ to denote the random vector $(X_1,X_2,\dots,X_n)$ and $x^n$ to denote the realization of such a random vector. For a real number $\alpha\in[0,1]$, we define $\bar{\alpha}=1-\alpha$. The binary entropy function is defined by $H_2(\alpha) = -\alpha\log_2(\alpha)-\bar{\alpha}\log_2(\bar{\alpha})$ with the convention of $0\log_2 0=0$. The probability mass function of $X$ is denoted by $P_X$, the conditional probability of $X$ given $Y$ is denoted by $P_{X|Y}$, and the joint distribution of $X$ and $Y$ is denoted by $P_{X,Y}$. 
The probability $\Pr[X=x]$ is denoted by $P_X(x)$. When the random variable is clear from the context, we write it in shorthand as $P(x)$. 

Let $P_Y$ and $T_Y$ be two discrete probability measures on the same probability space. The relative entropy between $P_Y$ and $T_Y$ is denoted by $D\left(P_Y\|T_Y\right)$. The conditional relative entropy is defined as $D(P_{Y|X}\|T_Y|P_X) = \mathbb{E}_X \left\{D(P_{Y|X}\|T_Y)\right\}$, where $\mathbb{E}_X[\cdot]$ denotes the expectation operator over $P_X$. We use the standard notation of directed information, as in \cite{Kramer98},
\begin{align}
    I(X^n\rightarrow Y^n) = \sum_{i=1}^n I(X^i;Y_i|Y^{i-1}),\nn
\end{align}
and causal conditioning
\begin{align*}
    P(y^n\|x^n) = \prod_{i=1}^n P(y_i|y^{i-1},x^i).
\end{align*}
When referring to causal conditioning particularized for deterministic functions, we employ the notation $f(x^n\|y^n)$, defined as
\begin{align}\label{eq:det_causal}
    f(x^n\|y^{n-1}) = \prod_{i=1}^n \mathbbm{1}\{x_i=f_i(x^{i-1},y^{i-1})\},
\end{align}
where $f_i:\mathcal{X}^{i-1}\times\mathcal{Y}^{i-1}\rightarrow \mathcal{X}$ are deterministic functions. Finally, the conditional causal conditioning is defined as
\begin{align*}
    P(y^n\|x^n|z) &= \prod_{i=1}^n P(y_i|y^{i-1},x^i,z).
\end{align*}

\subsection{Finite-state channels}
A FSC is defined by the triplet ($\cX\times\cS$, $P_{S^+,Y|X,S}$, $\cY\times\cS$), where $X$ is the channel input, $Y$ is the channel output, $S$ is the channel state at the beginning of the transmission, and $S^{+}$ is the channel state at the end of the transmission. The cardinalities $\cX$, $\cY$, and $\cS$ are assumed to be finite. The channel has the following probabilistic property
\begin{align}\label{eq:FSC}
    P(s_t,y_t|x^t,y^{t-1},s^{t-1},m) = P_{S^+,Y|X,S}(s_t,y_t|x_t,s_{t-1}),\;\;t=1,2,\dots,
\end{align}
for a given message $m$.

In this paper, we consider a communication setting with feedback as depicted in Fig. \ref{fig:FSC}. It is assumed that the initial state, $s_0$, is available to both the encoder and the decoder. The encoder has access to the message $M$, and the previous channel outputs. Accordingly, the encoder outputs $x_t$ as a function of $M$ and the channel outputs up to time $t-1$. The channel input $x_t$ then goes through a FSC and the resulting output
$y_t$ enters the decoder.

A FSC is \textit{strongly connected} if, for any states $s,s^{\prime}\in\mathcal{S}$, there exits an integer $T$ and an input distribution $\{P_{X_t|S_{t-1}}\}_{t=1}^{T}$ such that $\sum_{t=1}^T P_{S_t|S_0}(s|s^{\prime})>0$. We are using the definition of achievable rate and capacity as given in the book by Cover and Thomas \cite{cover}. The feedback capacity of a strongly connected FSC is stated in the following theorem.
\begin{theorem}[\!\cite{Kim08_feedback_directed}, Th. 1] \label{FSC_feedback_Capacity}
 The feedback capacity of a strongly connected FSC is
\begin{align*}
	\mathsf{C_{fb}} = \lim_{n\to\infty}\frac{1}{n}\max_{P(x^n\|y^{n-1})}I(X^n\rightarrow Y^n),
\end{align*}
for any initial state.
\end{theorem}
As can be seen above, the feedback capacity is expressed by a multi-letter formula. In the sequel, we present upper bounds for general FSCs, but these bounds are significantly simplified for two important classes of channels: 
\begin{itemize}
    \item Unifilar FSCs: for these channels, the channel state evolves according to a deterministic function of the channel input and output, and the previous channel state. That is, \eqref{eq:FSC} is simplified to:
    \begin{align}
        P(s_t,y_t|x_t,s_{t-1}) = \mathbbm{1}\{s_t=f(x_t,y_t,s_{t-1})\}P_{Y|X,S}(y_t|x_t,s_{t-1}),
    \end{align}
    where $f:\mathcal X\times \mathcal{Y}\times\mathcal{S}\to \mathcal S$.
    \item Finite-memory state channels: for these channels, the channel state depends on a finite number of past inputs and outputs. Specifically, 
       \begin{align}\label{eq:input_def}
        P(s_t,y_t|x^t,y^{t-1},s^{t-1}) = P(s_t|x_{t-k_1}^t,y_{t-k_2}^t)P_{Y|X,S}(y_t|x_t,s_{t-1}),
    \end{align}
    where $k_1$, $k_2$ are arbitrary non-negative finite integers.
\end{itemize}
In Section  \ref{subsection: FC_extention}, we show that the capacity of a finite-memory state channel can be computed as that of a new unifilar FSC, which is derived by reformulating the original channel. Specifically, we define the new channel state as $\tilde{S}_{t}\triangleq (X^t_{t-k_1},Y^t_{t-k_2})$, while the channel input and the channel output remain the same.

\section{Preliminaries} \label{sec: preliminaries}
In this section, we first introduce the dual capacity upper bound for memoryless channels. We then extend the technique to derive the duality bound for the case of directed information. Finally, we present an auxiliary tool called the $Q$-graph that will be used to produce test distributions.
\subsection{Dual capacity upper bound}
The dual capacity upper bound was first introduced in \cite{Topsoe67,Dual_capacity}. For a memoryless channel, the bound is given by the following theorem.
\begin{theorem}[\cite{Dual_capacity}, Th. $8.4$]
For any test distribution $T_Y$ and a memoryless channel $P_{Y|X}$, the capacity is upper bounded by
\begin{align}
    \mathsf{C} \le \max_{x\in\mathcal{X}} D(P_{Y|X=x}\|T_Y).
\end{align}
\end{theorem}
The bound is referred to as the \textit{duality upper bound} since it is derived from the dual capacity expression \cite{Csis62}, and it follows from the non-negativity of the KL divergence $D(P_Y ||T_Y)$. If the test distribution $T_Y$ is chosen to be the unique optimal output distribution $P_Y^*$, i.e., the output distribution induced by an optimal input distribution, the bound is tight. Therefore, to minimize the upper bound, one should carefully select $T_Y$ to be as close as possible to $P_Y^*$.

Following a similar methodology, we obtain the following result for the directed information.
\begin{theorem}[Duality UB for Directed Information]\label{Th: DB_DI}
For a fixed $P(x^n,y^n)$ and any test distribution $T(y^n)$,
\begin{align}
    I(X^n\rightarrow Y^n)\le\max D\left(P_{Y^n\|X^n=x^n}\|T_{Y^n}\right),
\end{align}
where the maximum is over $f(x^n\|y^{n-1})$ as defined in \eqref{eq:det_causal}.
\end{theorem}

It is worth noting that the directed information serves as the capacity expression objective for both feedback and non-feedback settings. The only difference lies in the optimization domain, where we optimize over causally conditioned input distributions $P(x^n|y^n)$ in the case of feedback and $P(x^n)$ in the case of no feedback. The duality bound also holds for non-feedback settings, but the optimization is done over sequences $x^n$ rather than feedback-dependent sequences $f(x^n|y^{n-1})$. The following proof illustrates this fact.
\begin{proof}[Proof of Theorem \ref{Th: DB_DI}]
Consider the following chain of inequalities
\begin{align}\label{eq: DB_DI}
    I(X^n\rightarrow Y^n) &= \sum_{x^n,y^n} P(x^n,y^n)\log_2\left(\frac{P(y^n\|x^n)}{P(y^n)}\right)\nn\\
    &= \sum_{x^n,y^n} P(x^n,y^n)\log_2\left(\frac{P(y^n\|x^n)T(y^n)}{P(y^n)T(y^n)}\right)\nn\\
    &= \sum_{x^n,y^n} P(x^n,y^n)\log_2\left(\frac{P(y^n\|x^n)}{T(y^n)}\right)-D(P_{Y^n}\|T_{Y^n})\nn\\
    &\stackrel{(a)}\le \sum_{x^n,y^n}P(x^n\|y^{n-1})P(y^n\|x^n)\log_2\left(\frac{P(y^n\|x^n)}{T(y^n)}\right)\nn\\
    &= \mathbb{E}_{X^{n}\|Y^{n-1}}\left[D\left(P_{Y^n\|X^n}\|T_{Y^n}\right)\right]\nn\\
    &\stackrel{(c)}\le \max_{f(x^n\|y^{n-1})} D\left(P_{Y^n\|X^n=x^n}\|T_{Y^n}\right),
\end{align}
where step $(a)$ follows from the non-negativity of the relative entropy and the fact that $P(x^n,y^n)=P(x^n\|y^{n-1})P(y^n\|x^n)$, and the maximum in step $(b)$ is taken over deterministic functions $f_i:\mathcal{X}^{i-1}\times\mathcal{Y}^{i-1}\rightarrow \mathcal{X}$ for $i=1,\dots,n$, such that $f(x^n\|y^{n-1}) = \prod_{i=1}^n \mathbbm{1}\{x_i=f_i(x^{i-1},y^{i-1})\}$.
\end{proof}
Note that Theorem \ref{Th: DB_DI} holds for any channel. As clarified, for channels with memory, one should consider choosing test distribution with memory \cite{Huleihel_Sabag_DB, OronBasharfeedback}. Markov test distributions are a standard choice in the literature \cite{Dual_Andrew_J,Duality_Kremer,Vontobel_Dual}, but it has been shown that the optimal outputs distribution for certain channels does not admit a Markovian structure of any finite order \cite{Sabag_UB_IT}. In this paper, we choose a test distribution that extends the Markov model to a variable-order Markov model using $Q$-graph.

\subsection{The $Q$-graph}
The $Q$-graph is a directed and connected graph defined on a finite set of nodes $\cQ$, with edges labeled with symbols from the channel output alphabet $\mathcal{Y}$. It possesses the property that the outgoing edges from each node are labeled with distinct symbols from $\cY$ (see Fig. \ref{fig:1Markov}). Consequently, the $Q$-graph can be utilized as a mapping of (any length) output sequences to the graph nodes by traversing the labeled edges. Specifically, for an initial node $q_0 \in \cQ$, we define a distinct mapping denoted by $\Phi_{q_0}:{\cY}^{*}\to \cQ$, where $\cY^*$ encompasses all finite-length sequences over $\cY$. To elaborate, $\Phi_{q_0}(y^t)$ signifies the node reached by traversing along the unique directed path of length $t$ labeled by $y^t = (y_1,y_2,\ldots,y_t)$ while originating from node $q_0$. For notational convenience, we frequently omit the subscript from $\Phi_{q_0}$ when there is no ambiguity.
Alternatively, the induced mapping can be expressed through a time-invariant function $\phi:\cQ\times\cY\to \cQ$, where a new node is computed from the previous node and the channel output.

Following \cite{Huleihel_Sabag_DB}, we use here \textit{graph-based test distributions} which extend the standard Markov test distributions to variable-order Markov models. For a fixed $Q$-graph, a graph-based test distribution, $T_{Y|Q}$, is defined by a collection of probability distributions $T_{Y|Q=q}$ on $\cY$ for each node $q \in \cQ$. It satisfies the following equation:
\begin{align}\label{eq:test_dist}
	T_{Y^n|Q_0}(y^n|q_0) = \prod_{t=1}^n T_{Y|Q}(y_t|q_{t-1}),
\end{align}
where $q_{t-1}=\Phi(y^{t-1})$. 
\begin{remark}
Note that a Markov model of any order is a special case of the variable-order Markov model. To elaborate, when considering a Markov model of order $k$, we construct a graph with $\mathcal{Y}^k$ nodes, where each node represents a tuple of $k$ channel outputs, and the edges are connected accordingly. Fig.~\ref{fig:1Markov} depicts an example of a first-order Markov model with $\mathcal{Y} = \{ 0, 1 \}$ and $k=1$. 
\end{remark}
\begin{figure}[tb]
\centering
    \psfrag{Q}[][][0.9]{$\;Y=1$}
    \psfrag{E}[][][0.9]{$\;\;Y=0$}
    \psfrag{L}[][][0.9]{$Q=2$}
    \psfrag{H}[][][0.9]{$Q=1$}
    \includegraphics[scale = 0.5]{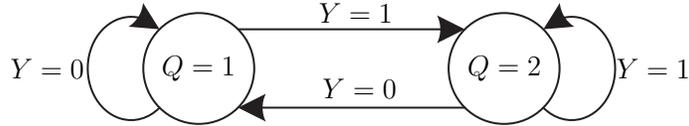}
    \caption{A $1$st-order Markov $Q$-graph for channel output alphabet $\mathcal{Y}=\{0,1\}$.}
    \label{fig:1Markov}
\end{figure}

\section{Main Results}\label{sec: main_results}
In this section, our main results are summarized. We begin by presenting the duality upper bound in Theorem \ref{Th: DB_DI} restricted to graph-based test distributions.
\begin{theorem}[Duality UB for FSCs using $Q$-graphs]\label{Th: DUB_Q}
For any graph-based test distribution $T_{Y|Q}$, the feedback capacity of a strongly connected FSC is bounded by
\begin{align}
    \mathsf{C_{fb}} \le \lim_{n\to\infty}\max_{f(x^n\|y^{n-1},s_0)}\min_{s_0,q_0}\frac{1}{n}\sum_{i=1}^n\mathbb{E}\left[D\left(\sum_{s_{i-1}}P(s_{i-1}|Y^{i-1},x^{i-1})P_{Y|X,S}(\cdot|x_i,s_{i-1})\Bigg{\|}T\left(\cdot|Q_{i-1}\right)\right)\right],\nn
\end{align}
where the joint distribution is defined by $P(x_i,y_i,s_i,q_i|x^{i-1},y^{i-1},s^{i-1},s_0,q_0)$ that factorizes as
\begin{align}
     \mathbbm{1}\{x_i=f_i(x^{i-1},y^{i-1},s_0)\}P_{Y|X,S}(y_i|x_i,s_{i-1})P_{S^+|X,Y,S}(s_i|x_i,y_i,s_{i-1})\mathbbm{1}\{q_i=\Phi_{q_0}(y^{i-1})\}.\nn
\end{align}
\end{theorem}
The proof of Theorem \ref{Th: DUB_Q} is shown in Appendix \ref{app: DUB_Q}. While our upper bound in Theorem \ref{Th: DUB_Q} is expressed as a multi-letter expression, we introduce the following theorem, stating our main result concerning the computability of the bound by formulating it as an MDP problem. 
\begin{theorem}\label{th: Main_mdp}
For any FSC, the dual capacity upper bound in Theorem \ref{Th: DUB_Q} can be formulated as an MDP that is presented in Table \ref{table: main1}. In particular, $\cal{P}(\mathcal{S})\times \mathcal{Q}$ is the state space, $\mathcal{X}$ is the action space, and $\mathcal{Y}$ is the disturbance space. For unifilar FSCs or finite-memory state channels, the MDP state simplifies to a singleton and takes values in the finite set $\mathcal{S}\times \mathcal{Q}$. 
\end{theorem}
The proof of Theorem \ref{th: Main_mdp} is provided in Section \ref{sec: DP_Main1}. We begin by presenting the MDP formulation of the dual capacity upper bound as stated in Theorem \ref{th: Main_mdp}. Subsequently, we proceed to demonstrate its validity as a well-defined MDP, establishing that its induced average reward equals the dual capacity upper bound. For unifilar FSCs and finite-memory state channels, the MDP formulation in Theorem \ref{th: Main_mdp} simplifies significantly to an MDP with finite states, actions, and disturbances. MDPs have been extensively studied in the field of optimization and control (e.g., \cite{Arapostathis,Bertsekas05,puterman1994}). When the MDP states, actions, and disturbances have finite spaces, the MDP can be solved relatively easily using standard DP algorithms such as the value iteration and the policy iteration. In particular, these algorithms are used to evaluate the MDP and to extract a conjectured solution. The conjectured solution is then proved using the Bellman equation, leading to the simplification of the involved upper bound in Theorem \ref{Th: DUB_Q} into explicit expressions.

In Section~\ref{sec:examples}, we demonstrate that the duality bound provides simple upper bounds when the test distribution is chosen correctly. Specifically, we apply the developed framework to several FSCs, resulting in novel analytical upper bounds on their capacity. The results are summarized as follows:
\begin{itemize}
\item We investigate a wide family of NOST channels and derive their capacity as a simple closed-form analytical expression, as shown in Theorem \ref{th: NOST}.
\item We establish new upper bounds on the feedback capacity of the N-Ising channel. The bounds are provided in Theorem \ref{Th: N-Ising} and Theorem \ref{Th: N-Ising_Q4}.
\end{itemize}

Additionally, in Section \ref{sec: Q_graph_exploration}, we present two methods based on convex optimization and RL \cite{OronBasharfeedback,Permuter06_trapdoor_submit} to extract $Q$-graphs with good performance. Although these approaches were initially proposed for unifilar FSCs, we demonstrate in Section \ref{sec: Q_graph_exploration} that they can be adapted to finite-memory state channels by reformulating the channel to a new unifilar FSC, with its capacity being the same as that of the original channel.

\section{Upper Bounds via MDP} \label{sec: DP_Main} 
In this section, we first introduce our MDP framework and the Bellman equation.
Then, for a fixed graph-based test distribution, we formally present the MDP formulation in Theorem \ref{th: Main_mdp} of the dual capacity upper bound for FSCs with feedback. 

\subsection{MDP and the Bellman equation}
MDP is a mathematical framework for modeling sequential decision-making problems involving uncertain outcomes that depend on the system's current state. We consider an MDP scenario encompassing a state space denoted as $\mathcal{Z}$, an action space denoted as $\mathcal{U}$, and a disturbance space denoted as $\mathcal{W}$. First, the initial state $z_0$, is drawn randomly from the distribution $P_Z$. At each time step $t$, the system resides in a state $z_{t-1}\in\mathcal{Z}$, wherein the decision-maker selects an action $u_t\in\mathcal{U}$ and a disturbance $w_t\in\mathcal{W}$ is drawn from the conditional distribution $P_w(\cdot|z_{t-1},u_t)$. The state $z_t$ then evolves according to the transition function $F:\mathcal{Z}\times\mathcal{U}\times\mathcal{W}\to\mathcal{Z}$, resulting in $z_{t}=F(z_{t-1},u_{t},w_{t})$.

To determine the action $u_t$, the decision-maker relies on the function $\mu_t$, which maps histories $h_t=(z_0,w_0,\dots,w_{t-1})$ to corresponding actions, denoted as $u_t = \mu_t(h_t)$. Our objective, given a policy $\pi=\{\mu_1,\mu_2,...\}$ and a bounded reward function $g:\mathcal{Z}\times\mathcal{U}\to\mathbb{R}$, is to maximize the average reward over an infinite time horizon. The average reward achieved by policy $\pi$ is defined as 
$\rho_\pi = \liminf_{n\to\infty}\frac{1}{n}\mathbb{E}_\pi\left[\sum_{t=0}^{n-1}g\left(Z_t,\mu_{t+1}(z_0)\right)\right]$. Accordingly, the optimal average reward is given by $\rho^*=\sup_{\pi}\rho_\pi$.

The Bellman equation provides an alternative characterization for the optimal average reward in MDPs. The following theorem encapsulates the Bellman equation for our formulation.
\begin{theorem}[Bellman equation, \cite{Arapos93_average_cose_survey}]\label{Theorem:Bellman}
If a scalar $\rho\in\mathbb{R}$ and a bounded function $h:\mathcal{Z}\rightarrow\mathbb{R}$ satisfy
\begin{align}\label{eq:Bellman}
    \rho+h(z) = \max_{u\in\mathcal{U}}\pr{g\pr{z,u}+\int P_w(dw|z,u)h\pr{F\pr{z,u,w}}},\;\; \forall z\in\mathcal{Z}
\end{align}
then $\rho^{*}=\rho$.
\end{theorem}
For MDPs with finite states, actions, and disturbances, the process of finding the function $h(\cdot)$ and the scalar $\rho^*$ in the Bellman equation simplifies significantly. Specifically, the optimal policy $\pi^*$ can be deduced utilizing DP algorithms like value iteration, policy iteration, or RL techniques. Once an optimal policy is established, the Bellman equation transforms into a finite set of linear equations that can be easily resolved. In the subsequent section, we delve into the MDP formulation of the dual upper bound and elucidate that, for unifilar FSCs and finite-memory state channels, the MDP formulation consists of finite states, actions, and disturbances.

\subsection{MDP formulation of the Dual Upper Bound (Proof of Theorem \ref{th: Main_mdp})} \label{sec: DP_Main1} 
In the following, we formally present the MDP formulation of the dual capacity upper bound for FSCs with feedback stated in Theorem \ref{th: Main_mdp}. 
\begin{table}[t]
\caption{MDP Formulation}
\label{table: main1}
\begin{center}
\scalebox{1.2}{
 \begin{tabular}{|c | c |} 
 \hline
MDP notations & Upper bound on capacity \\ [0.5ex] 
 \hline\hline
MDP state & $(\beta_{t-1},q_{t-1})\triangleq(P(S_{t-1}|x^{t-1},y^{t-1},s_0),q_{t-1})$ \\ [0.3ex]
 \hline
Action & $x_t$ \\[0.3ex]
 \hline
Disturbance & $y_t$ \\[0.3ex]
 \hline
The reward & $D\left(\sum_{s_{t-1}}\beta_{t-1}(s_{t-1}) P_{Y|X,S}(\cdot|x_t,s_{t-1})\middle\|T_{Y|Q}(\cdot|q_{t-1})\right)$\\[0.3ex]
 \hline
\end{tabular} }
\end{center}
\end{table}
Fix a $Q$-graph and a corresponding graph-based test distribution $T_{Y|Q}$. The MDP state at time $t$ is defined as $z_{t-1}\triangleq (\beta_{t-1},q_{t-1})$, where $\beta_{t-1}\triangleq \left[\beta_{t-1}(0)\cdots,\beta_{t-1}(|\mathcal{S}|-1)\right]$ and $\beta_{t-1}(s)\triangleq P(S_{t-1}=s|x^{t-1},y^{t-1},s_0)$.
The action is defined as the channel input $u_t \triangleq x_t$, and the disturbance is defined as the channel output $w_t\triangleq y_t$. The reward function is defined as follows:
\begin{align} \label{reward_1}
	g(z_{t-1},x_t) \triangleq D\left(\sum_{s_{t-1}}\beta_{t-1}(s_{t-1}) P_{Y|X,S}(\cdot|x_t,s_{t-1})\middle\|T_{Y|Q}(\cdot|q_{t-1})\right).
\end{align}
The MDP formulation is summarized in Table \ref{table: main1}. The induced average reward in the infinite horizon regime of this MDP is given by
\begin{align} \label{eq: optimal_reward_1}
    &\rho^*=\sup\liminf_{n\to\infty}\min_{s_0,q_0}\frac{1}{n}\sum_{t=1}^{n}\mathbb{E}\left[D\left(\sum_{s_{t-1}}\beta_{t-1}(s_{t-1})P_{Y|X,S}(\cdot|x_t,s_{t-1})\middle\|T_{Y|Q}(\cdot|Q_{t-1})\right)\right], 
\end{align}
where the supremum is over all deterministic functions $\{f_i:\mathcal{X}^{i-1}\times\mathcal{Y}^{i-1}\rightarrow \mathcal{X}\}_{i\ge 1}$. The following theorem summarizes the relationship between the upper bound in Theorem \ref{Th: DUB_Q} and $\rho^\ast$.
\begin{theorem}\label{theorem: formulation}
The upper bound in Theorem \ref{Th: DUB_Q} is equal to the optimal average reward in \eqref{eq: optimal_reward_1}. That is, the capacity is upper bounded by $\rho^*$.
\end{theorem}
The proof of Theorem \ref{theorem: formulation} is given in Appendix~\ref{app: formulation}. As part of the proof, we show that the formulation above constitutes a valid MDP.

\subsection{Simplified MDP formulation}\label{sec:simplified_mdp}
In the previous section, we introduced an MDP formulation of the dual upper bound that holds for any FSC. In this section, we highlight that the proposed formulation simplifies significantly in the case of unifilar FSCs and finite-memory state channels. Specifically, we demonstrate that, for these cases, the MDP state space is finite.

Recall the MDP state is $z_{t-1}\triangleq (\beta_{t-1},q_{t-1})$, where $\beta_{t-1}\triangleq P(S_{t-1}|x^{t-1},y^{t-1},s_0)$. For unifilar FSCs, note that
\begin{align}
    \beta_{t-1}(s_{t-1})&\triangleq P(s_{t-1}|x^{t-1},y^{t-1},s_0)\nn\\
    &=\mathbbm{1}\{s_{t-1}=f(x_{t-1},y_{t-1},s_{t-2})\}.\nn
\end{align}
Consequently, for unifilar FSCs, the vector $\beta_{t-1}$ is equal to $1$ at the corresponding coordinate where $s_{t-1}=f(x_{t-1},y_{t-1},s_{t-2})$ and is $0$ otherwise. Thus, the MDP state can be represented as $z_{t-1}=(s_{t-1},q_{t-1})$, which takes values in $\mathcal{S}\times\mathcal{Q}$.

For finite-memory state channels, we observe that
\begin{align}\label{eq: finite_memo_state}
        \beta_{t-1}(s_{t-1})&\triangleq P(s_{t-1}|x^{t-1},y^{t-1},s_0)\nn\\
        &=P(s_{t-1}|x_{t-k_1}^{t-1},y_{t-k_2}^{t-1}),
\end{align}
where $k_1,k_2>0$ are finite and do not depend on $t$. Therefore, according to \eqref{eq: finite_memo_state}, the MDP state space is finite and at most of size $|\mathcal{X}|^{k_1}\cdot|\mathcal{Y}|^{k_2}$.

\begin{remark}\label{remark: bellman_sol}
In this simplified MDP formulation, the state space, actions, and disturbances are all finite. Consequently, solving the MDP numerically and finding the optimal policy becomes relatively straightforward. Specifically, standard DP and RL algorithms converge quickly to the optimal policy. Moreover, with an optimal policy in hand, the Bellman equation can be represented as a finite set of linear equations due to the finite alphabets. These equations can be directly solved to obtain the scalar $\rho^*$ and the function $h$ that satisfy the Bellman equation. Consequently, this provides an explicit expression for the upper bound on capacity (rather than a numerical approximation). The upper bound on capacity is determined by the previously found $\rho^*$, i.e., $\mathsf{C_{fb}}\leq \rho^*$.
\end{remark}

\section{$Q$-graph Exploration}\label{sec: Q_graph_exploration}
The main challenge so far lies in finding a suitable $Q$-graph that its corresponding graph-based test distribution will result in a good or tight upper bound when optimized. Once a $Q$-graph is selected, our methodology allows us to evaluate the upper bound and derive its analytical expression by solving the MDP problem. In this section, we introduce two approaches, proposed in \cite{OronBasharfeedback, aharoni2022feedback}, to address this challenge. Although these approaches hold for unifilar FSCs, we show in Section \ref{subsection: FC_extention} that they can be adapted to finite-memory state channels.

\subsection{Exploration via the $Q$-graph bounds}\label{subsection: FC_Qbounds}
In \cite{OronBasharfeedback}, the authors introduced optimization algorithms to compute upper and lower bounds on the capacity of unifilar FSCs with feedback. For a fixed $Q$-graph, these algorithms were based on single-letter bounds derived in \cite{Sabag_UB_IT}. To identify a suitable $Q$-graph, they suggested conducting an exhaustive search over all valid $Q$-graphs\footnote{A valid $Q$-graph is a directed graph that is aperiodic, i.e., connected and has period 1.}. However, the computational complexity of this exhaustive search increases exponentially with the graph size and the cardinality of the channel output. For example, when $|\mathcal{Y}|=2$ and we search over all valid $Q$-graphs of size 6, there are 655,424 different graphs. Consequently, this approach is primarily suitable for exploring small-sized $Q$-graphs. 

\begin{remark}
In \cite{OronBasharfeedback}, it is demonstrated that the $Q$-graph upper bound from \cite{Sabag_UB_IT} can be formulated as a standard convex optimization problem. However, the analytical computation of the bound remains challenging due to the need to verify the KKT conditions. The approach introduced in this paper offers an alternative and significantly simpler approach to derive analytical upper bounds, even for channels with large alphabets. However, it is worth noting that, for a fixed $Q$-graph, our methodology mandates optimizing the corresponding graph-based test distribution to obtain meaningful bounds. In contrast, the $Q$-graph upper bound does not necessitate such optimization as it does not involve the use of a test distribution.
\end{remark}

\subsection{Exploration via RL}\label{subsection: FC_RL}
In \cite{Permuter06_trapdoor_submit}, the feedback capacity of unifilar FSCs was formulated as an MDP problem. Then, in \cite{aharoni2022feedback}, the authors presented an RL-based algorithm to compute the feedback capacity based on this MDP formulation. The key advantage of the algorithm is its ability to evaluate the capacity even for channels with large alphabets. While this methodology is primarily employed for numerical evaluation of the capacity, in some cases, it can also be utilized to extract optimal or near-optimal $Q$-graphs.
Specifically, the algorithm generates a histogram of the MDP states that are visited under an estimated optimal policy. If the resulting histogram of the MDP states is discrete, indicating that only a finite number of MDP states that are visited, a $Q$-graph can be extracted. In this case, each visited MDP state serves as a node in the $Q$-graph, and the labeled edges capture the evolution of the MDP states as a function of the channel outputs. For additional details, we refer the reader to \cite{Permuter06_trapdoor_submit}.

\subsection{Extension to finite-memory state channels}\label{subsection: FC_extention}
In this section, we demonstrate that the two approaches in Sections \ref{subsection: FC_Qbounds} and \ref{subsection: FC_RL} can be adapted to finite-memory state channels. The idea is to show that the capacity of a finite-memory state channel can be computed as the capacity of a new unifilar FSC obtained by reformulating the original channel.

Consider a finite-memory state channel and define the following transformation:
\begin{itemize}
    \item The channel state is $\tilde{S}_{t}\triangleq (X^t_{t-k_1},Y^t_{t-k_2})$.
    \item The channel input and the channel output remain the same.
\end{itemize}
We show that the above transformation defines a new unifilar FSC with a transition kernel $P_{Y,\tilde{S}^+|X,\tilde{S}}$. In other words, the new channel follows the time-invariant Markov property of FSCs defined in Eq. \eqref{eq:FSC}. Further, we establish the relationship between the capacity of the original channel and its transformed version through the following theorem.
\begin{theorem}\label{th: finite_memo_as_unifilar}
The capacity of any finite-memory state channel is equal to the capacity of its transformed FSC $P_{Y,\tilde{S}^+|X,\tilde{S}}$. Further, the transformed channel is a unifilar FSC.
\end{theorem}

\begin{proof}[Proof of Theorem \ref{th: finite_memo_as_unifilar}]
We first show that the new channel is a unifilar FSC. The proof consists of the following three steps.
\begin{enumerate}
    \item Conditioned on the previous channel state $\tilde{S}_{t-1}$ and the channel input $X_t$, the channel output $Y_t$ is independent of any past states, inputs, and outputs. In particular, note that 
    \begin{align} \label{eq: new_channel_law}
    P(y_t|x^t,y^{t-1},\tilde{s}^{t-1}) &=\sum_{s_{t-1}} P(s_{t-1}|x^t,y^{t-1},\tilde{s}^{t-1})P(y_t|x^t,y^{t-1},\tilde{s}^{t-1},s_{t-1})\nn\\
    &\stackrel{(a)}= \sum_{s_{t-1}} P(s_{t-1}|x^{t},y^{t-1},\tilde{s}^{t-1})P_{Y|X,S}(y_t|x_t,s_{t-1})\nn\\
    &\stackrel{(b)}= \sum_{s_{t-1}} P(s_{t-1}|\tilde{s}_{t-1})P_{Y|X,S}(y_t|x_t,s_{t-1}),
    \end{align}
    where $(a)$ follows by the Markov property of the original channel, which implies the Markov chain $Y_t-(X_t,S_{t-1})-(X^{t-1},Y^{t-1},\tilde{S}^{t-1})$, and $(b)$ follows by the fact that $\tilde{s}_{t-1}$ includes $(x_{t-k_1-1}^{t-1},y_{t-k_2-1}^{t-1})$. Hence, according to \eqref{eq: new_channel_law}, the required property holds.
    \item Next, we show that the unifilar property holds. Since $\tilde{s}_t=\left(x_{t-k_1}^t,y_{t-k_2}^t\right)$, it follows directly that there exits a deterministic function $\tilde{f}:\mathcal{X}\times \mathcal{Y}\times\tilde{\mathcal{S}}\to \tilde{\mathcal{S}}$ such that $\tilde{s}_t=\tilde{f}(x_t,y_t,\tilde{s}_{t-1})$. This implies that the unifilar property indeed holds.
    \item Recall that both $\mathcal{X}$ and $\mathcal{Y}$ are finite cardinalities. Accordingly, we directly deduce that the cardinality of the new channel state is finite as well.
\end{enumerate}

Finally, we note that the capacity of the new unifilar FSC is equal to the capacity of the original channel. Specifically, given an input sequence, the corresponding outputs of the new channel are drawn according to the statistics of the original channel model. Furthermore, the maximization domain remains unchanged as we still maximize over $P(x^n\|y^{n-1})$.
\end{proof}

\section{Examples}\label{sec:examples}
In this section, we present several examples of FSCs. For all examples, both the channel input and the channel state take values from the binary alphabet, i.e., $\mathcal{S}=\mathcal{X} = \{0,1\}$. The studied examples illustrate the simplicity of deriving analytical upper bounds using our proposed methodology.

\subsection{The NOST Channels}\label{subsec:nost}
The NOST channels, studied in \cite{shemuel2022feedback}, are an extension of the well-known Previous  Output is the STate (POST) channels \cite{POSTchannel}. For these channels, the state of the channel depends stochastically on the previous channel output according to $P(s_t|y_t)$. We focus on a broad family of NOST channels that is defined as follows. At time $t$, if $s_{t-1} = 0$, the channel behaves according to a $Z$ topology with a parameter of $0.5$, while if $s_{t-1} = 1$, it behaves according to an $S$ topology with a parameter of $0.5$. The state evolution $P(s_t|y_t)$ follows a binary symmetric channel (BSC) topology with a transition parameter of $\epsilon$.

\begin{figure}[b]
\centering
    \includegraphics[scale = 0.6]{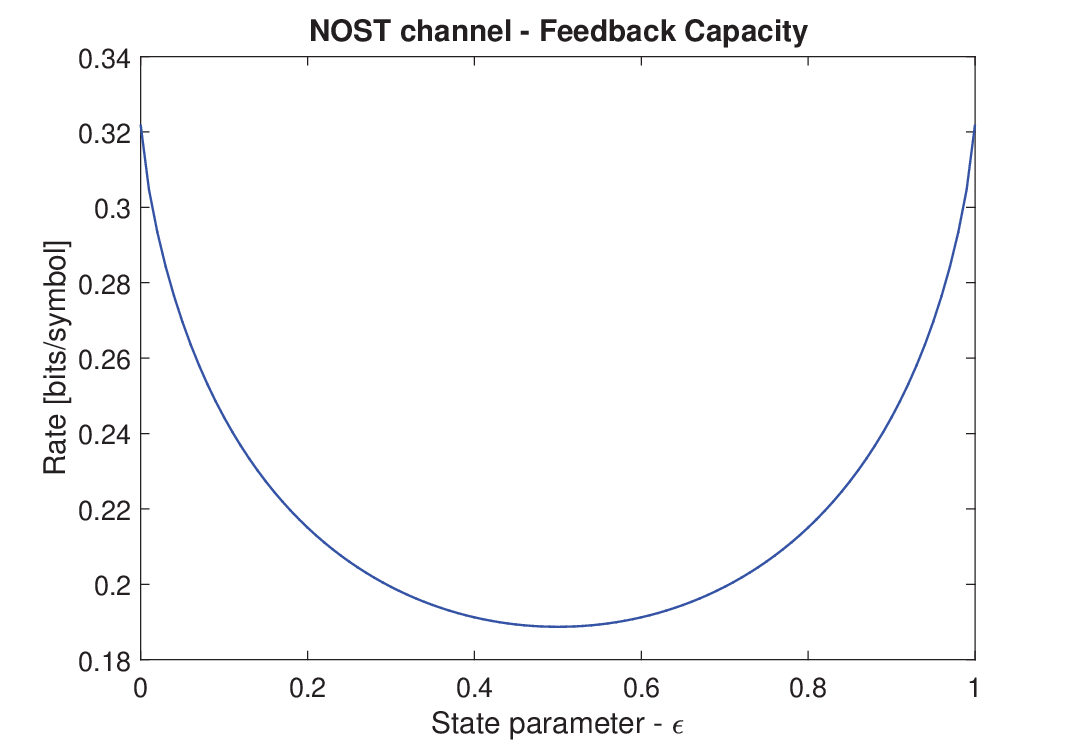}
    \caption{The feedback capacity of the NOST channel as a function of the state parameter $\epsilon$.}
    \label{fig:NOST_capacity}
\end{figure}
The feedback capacity of POST channels was derived in \cite{POSTchannel}. 
In a subsequent extension of this work \cite{shemuel2022feedback}, the authors determined the capacity of NOST channels as a single-letter optimization problem over mutual information. 
Our contribution is a new result regarding the capacity of NOST channels with BSC state evolution, as defined above. We establish their capacity through a simple closed-form formula, as demonstrated in the following theorem. To illustrate, Figure \ref{fig:NOST_capacity} provides a numerical evaluation of the capacity.
\begin{theorem} \label{th: NOST}
For any $\epsilon\in[0,1]$, the capacity of the NOST channels is given by
\begin{align*}
\mathsf{C_{NOST}}(\epsilon) = \frac{1}{2}\log_2\left(\frac{1}{4}\left(\frac{\epsilon}{1-a}\right)^{\epsilon}\left(\frac{1+\bar{\epsilon}}{a}\right)^{1+\bar{\epsilon}}\right),
\end{align*}
where
\begin{align*}
    a = \frac{\epsilon^\epsilon(1+\bar{\epsilon})^{1+\bar{\epsilon}}}{\epsilon^\epsilon(1+\bar{\epsilon})^{1+\bar{\epsilon}} + \bar{\epsilon}^{\bar{\epsilon}}(1+\epsilon)^{1+\epsilon}}.
\end{align*}
\end{theorem}
The proof of Theorem \ref{th: NOST} is provided at the end of this section, where both upper and lower bounds are presented and shown to be equal, thereby establishing the capacity. These bounds were obtained using a first-order Markov $Q$-graph. In fact, it can be concluded from \cite{shemuel2022feedback} that a first-order Markov $Q$-graph is sufficient to solve any instance of NOST channels, where the dependence of the channel state on the previous channel output does not necessarily have to follow a BSC topology.

\begin{remark}
As mentioned earlier, POST channels are instances of NOST channels, where the channel state is a deterministic function of the output, that is, $P(s_t|y_t)=\mathbbm{1}\{{s_t=y_t}\}$. Notably, when $\epsilon=0$, we obtain a POST channel instance. In this case, the capacity is $\mathsf{C_{NOST}}(0)=\log_2(5/4)$, precisely matching the capacity outlined in \cite{POSTchannel} for this specific POST channel.
\end{remark}
\begin{proof}[Proof of Theorem \ref{th: NOST}]
The proof of Theorem \ref{th: NOST}, concerning the capacity of the NOST channel, consists of two parts. First, we derive a tight upper bound on its capacity. Following that, we provide a corresponding lower bound to demonstrate the tightness of the upper bound.
The upper bound is shown in this section to elucidate the process of deriving analytical upper bounds using our methodology. Conversely, the lower bound is detailed in Appendix \ref{app: NOST_LB} to maintain the flow of the presentation.

To derive the upper bound, we solve the Bellman equation in \eqref{eq:Bellman}. Consider the first-order Markov $Q$-graph in Fig. \ref{fig:1Markov}. For any state parameter $\epsilon\in[0,1]$, we define the following graph-based test distribution:
\begin{align*}
    T_{Y|Q}(0|0)=T_{Y|Q}(1|1)=a,
\end{align*}
where
\begin{align}\label{eq: opt_a}
    a = \frac{\epsilon^\epsilon(1+\bar{\epsilon})^{1+\bar{\epsilon}}}{\epsilon^\epsilon(1+\bar{\epsilon})^{1+\bar{\epsilon}} + \bar{\epsilon}^{\bar{\epsilon}}(1+\epsilon)^{1+\epsilon}}.
\end{align}

Since the NOST channel is a finite-memory state channel, the MDP state is defined as $z_t=(\beta_t,q_t)$, corresponding to the MDP formulation in Section \ref{sec: DP_Main}, where
\begin{align*}
    \beta_{t}=P(S_t|x_t) = 
    \begin{cases} [1-\epsilon,\epsilon], & x_t=0, \\
                  [\epsilon,1-\epsilon], & x_t=1.
    \end{cases}
\end{align*}

By iterating the value iteration algorithm, one can deduce a conjectured solution\footnote{Further reading on the numerical evaluation of MDPs can be found in \cite{Permuter06_trapdoor_submit,Sabag_BEC}.} for the value function $h(z)$ and $\rho^*$, which are required for the solution of the Bellman equation. Specifically, define
\begin{align}\label{eq:rho_nost}
    \rho^*=\frac{1}{2}\log_2\left(\frac{1}{4}\left(\frac{\epsilon}{1-a}\right)^{\epsilon}\left(\frac{1+\bar{\epsilon}}{a}\right)^{1+\bar{\epsilon}}\right).
\end{align}
Further, define the value function as follows:
\begin{align}\label{eq: NOST_h_func}
    h([1-\epsilon,\epsilon],1)&=h([\epsilon,1-\epsilon],2)=0 \nn\\
    h([1-\epsilon,\epsilon],2)&=h([\epsilon,1-\epsilon],1)=\log_2\left(\frac{\bar{a}a^\epsilon}{a\bar{a}^\epsilon}\right).
\end{align}

In the following, we show that $\rho^*$ in \eqref{eq:rho_nost} and $h(z)$ in \eqref{eq: NOST_h_func} solve the Bellman equation. We begin with computing the MDP operator at the state $(\beta=[1-\epsilon,\epsilon],q=1)$. Specifically, it is the maximum (on the right-hand side of the Bellman equation in \eqref{eq:Bellman}) between two terms:
\begin{align}\label{eq: NOST_bellman}
    x=0:&D\left(\left[1-0.5\epsilon,0.5\epsilon\right]\big{\|}\left[a,1-a\right] \right) + \left(1-\frac{\epsilon}{2}\right)\cdot h([1-\epsilon,\epsilon],1) + \frac{\epsilon}{2}\cdot h([\epsilon,1-\epsilon],2)\nn\\
    x=1:&D\left(\left[0.5(1-\epsilon),0.5(1+\epsilon)\right]\big{\|}\left[a,1-a\right] \right) + \frac{1-\epsilon}{2}\cdot h([1-\epsilon,\epsilon],1) + \frac{1+\epsilon}{2}\cdot h([\epsilon,1-\epsilon],2).
\end{align}
After substitution, we can express \eqref{eq: NOST_bellman} as
\begin{align}\label{eq: NOST_bellman2}
    x=0:&\;\;(1-0.5\epsilon)\log_2\left(\frac{1-0.5\epsilon}{a}\right)+0.5\epsilon\log_2\left(\frac{\epsilon}{2\bar{a}}\right)
    \nn\\
    x=1:&\;\;\frac{\bar{\epsilon}}{2}\log_2\left(\frac{\bar{\epsilon}}{2a}\right)+\frac{1+\epsilon}{2}\log_2\left(\frac{1+\epsilon}{2\bar{a}}\right).
\end{align}
By substituting $a$ into \eqref{eq: NOST_bellman2} and comparing the expressions for $x=0$ and $x=1$, it can be shown that both actions yield the same result. Furthermore, it can be directly verified that \eqref{eq: NOST_bellman2} equals $\rho^*$ by comparing $\rho^*$ with the expression for $x=0$. On the other hand, we note that the left-hand side of the Bellman equation is also $\rho^*+h([1-\epsilon,\epsilon],1)=\rho^*$. Accordingly, the Bellman equation holds for the MDP state $z=([1-\epsilon,\epsilon],1)$. The verification for the remaining MDP states can be conducted similarly.
\end{proof}

\subsection{The Noisy Ising channel} 
The Ising channel, introduced by Berger and Bonomi in 1990 \cite{Berger90IsingChannel}, models a channel with intersymbol interference. The channel is defined as follows:
\begin{align}\label{eq: Ising_channel_def}
    Y_t = \begin{cases}
        X_t & \textit{, w.p. 0.5},\\
        X_{t-1} &\textit{, w.p. 0.5}.
    \end{cases}
\end{align}
The channel state $s_t$ is equal to the channel input $x_t$, making this channel a unifilar FSC. The feedback capacity of this channel has been determined in \cite{Ising_channel}. Further, in a recent advancement to RL algorithms \cite{aharoni2022feedback}, the capacity of the Ising channel with alphabet size $|\mathcal{X}|\le 8$ has been established. The duality bound, derived in the current paper, is the main tool used to prove the converse.

We study here the N-Ising channel, a generalized version of the Ising channel where the channel state is stochastically determined by the last input. Specifically, the channel is given by Eq. \eqref{eq: Ising_channel_def}, but the new channel state is randomly chosen according to a distribution $P(s_i|x_i)$ resembling a BSC topology with a transition probability $\epsilon$. We refer to these channels as the Noisy Ising (N-Ising) channels. Note that the N-Ising channel is not a unifilar FSC, but it is a finite-memory state channel. In the following theorem, we provide a simple upper bound on its capacity.

\begin{figure}[t]
\centering
    \includegraphics[scale = 0.6]{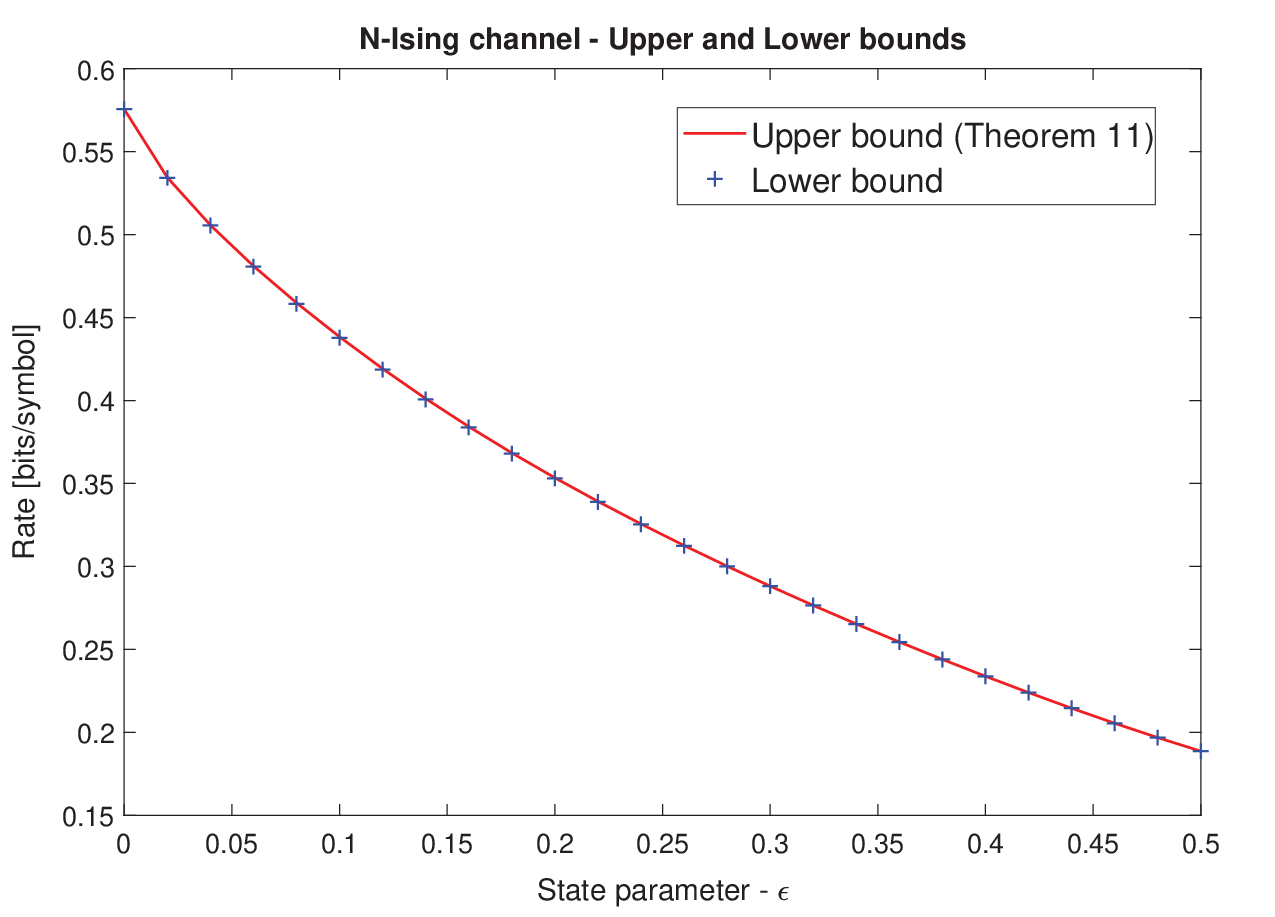}
    \caption{Bounds on the capacity of the N-Ising channel. The upper bound is obtained by a single $Q$-graph of size $4$ (Theorem \ref{Th: N-Ising_Q4}). The lower bound is evaluated with two different $Q$-graphs of size $10$ and $12$.}
    \label{fig:N-Ising_bounds}
\end{figure}
\begin{theorem}\label{Th: N-Ising}
For any $\epsilon\in[0,0.5)$, the capacity of the N-Ising channel is upper bounded by 
        \begin{align*} 
            \mathsf{C_{N-Ising}}(\epsilon)\leq  \frac{1}{2+4\bar{\epsilon}}\cdot\log_2\left(\frac{16^\epsilon\epsilon^{\epsilon\bar{\epsilon}}\left(\bar{\epsilon}(1+\bar{\epsilon})\right)^{\epsilon^2-3\epsilon+2}}{64\bar{a}^2(a\bar{a})^{2\bar{\epsilon}}(1+\epsilon)^{\epsilon^2-\epsilon-2}}\right),
        \end{align*}
        where
        \begin{align*}
            a = \frac{\epsilon^{\epsilon\gamma}}{\epsilon^{\epsilon\gamma}+(\bar{\epsilon}^{\bar{\epsilon}}(1+\bar{\epsilon})^{\epsilon-2})^\gamma (1+\epsilon)^{\gamma(1+\epsilon)}},\;\; \gamma=\frac{1}{2\epsilon^2-3\epsilon+2}.
        \end{align*}
\end{theorem}
The proof of Theorem \ref{Th: N-Ising} is provided in Appendix \ref{app: N-Ising}. The upper bound is derived using a first-order Markov $Q$-graph. Transforming the N-Ising channel into a new unifilar FSC, as outlined in Section \ref{subsection: FC_extention}, yields a symmetric channel. The inherent symmetry of the channel leads to the observation that, for any $\epsilon\in(0.5,1]$, we have $\mathsf{C_{N-Ising}}(\epsilon) = \mathsf{C_{N-Ising}}(1-\epsilon)$. Additionally, for $\epsilon=0.5$, the transformed channel is a BSC with parameter $0.25$ for any channel state. Thus, for $\epsilon=0.5$, the capacity is $\mathsf{C_{N-Ising}}(0.5)=1-H_2(0.25)$.

In the following theorem, we present an additional upper bound that is at least as good as the upper bound in Theorem \ref{Th: N-Ising}, and is tighter in the case of small values of the state parameter $\epsilon$. The difference between the upper bounds in Theorems \ref{Th: N-Ising} and \ref{Th: N-Ising_Q4} has an order of $\sim 10^{-2}$. Considering the simplicity of the expression and the absence of optimization procedures in Theorem \ref{Th: N-Ising}, we present it as well.

\begin{theorem}\label{Th: N-Ising_Q4}
For any $\epsilon\in[0,0.5)$, the capacity of the N-Ising channel is upper bounded by
\begin{align*} 
        \mathsf{C_{N-Ising}}(\epsilon)\leq  \min\frac{1}{2(1+2\bar{\epsilon})(2+\epsilon)}\cdot\log_2\left(\frac{2^{4\epsilon^2+2\epsilon-12}\bar{a}^{4\epsilon^2-6\epsilon-4}(\epsilon^2-3\epsilon+2)^{\epsilon^3-\epsilon^2-4\epsilon+4}}{a^{4\bar{\epsilon}}b^{2\epsilon\bar{\epsilon}}(1+\epsilon)^{\epsilon^3+\epsilon^2-4\epsilon-4}\epsilon^{\epsilon^3+\epsilon^2-2\epsilon}\bar{b}^{2\epsilon^2-6\epsilon+4}}\right), 
\end{align*}
where the minimum is over all $(a,b)\in[0,1]^2$ that satisfy:
\begin{align} 
         1 \leq \frac{b^{7\epsilon^3-2\epsilon^4+4\epsilon^2-18\epsilon+12}\bar{b}^{2\epsilon^4-7\epsilon^3+16\epsilon-8}\epsilon^{\epsilon^3-2\epsilon^2-8\epsilon}\bar{\epsilon}^{\epsilon^3-3\epsilon^2-6\epsilon+8}}{a^{2\epsilon-4\epsilon^2-4}\bar{a}^{8\epsilon^2-4\epsilon+8}(1+\epsilon)^{\epsilon^3-\epsilon^2-10\epsilon-8}(2-\epsilon)^{\epsilon^3-4\epsilon^2-4\epsilon+16}}. \nn
\end{align}
\end{theorem}
The proof of Theorem \ref{Th: N-Ising_Q4} is shown in Appendix \ref{app: N-Ising}. This upper bound is obtained by using a unique $Q$-graph of size $4$, that is given within the proof of the theorem.

In Fig. \ref{fig:N-Ising_bounds}, we present a numerical evaluation of upper and lower bounds on the capacity of the N-Ising channel as a function of the state parameter $\epsilon$. To evaluate the performance of the upper bounds, Fig. \ref{fig:N-Ising_bounds} compare the upper bound in Theorem \ref{Th: N-Ising_Q4} with a lower bound obtained using the $Q$-graph lower bound \cite{OronBasharfeedback}. As explained in Section \ref{subsection: FC_Qbounds}, this lower bound is applicable only to unifilar FSCs. Therefore, we first reformulate the N-Ising channel as a new unifilar FSC, as described in Section \ref{subsection: FC_extention}, and then apply the $Q$-graph lower bound on the reformulated channel. 

We evaluated the lower bound using two different $Q$-graphs of sizes $10$ and $12$, which were obtained through the RL methodology described in Section \ref{subsection: FC_RL}. In particular, to illustrate this, Fig. \ref{fig:N-Ising_Qgraph} presents two histograms, each corresponding to a different value of the state parameter. These histograms depict the MDP states visited under an estimated optimal policy learned by RL. Following this, we perform a quantization process to extract finite-sized $Q$-graphs. It is worth noting that this quantization may lead to sub-optimal $Q$-graphs, which may not result in a tight bound. However, in many cases, it offers insights into the optimal $Q$-graph or leads to a $Q$-graph that provides a reasonably good bound. In the case of the N-Ising channel, as can be seen in Fig. \ref{fig:N-Ising_bounds}, the RL algorithm indeed converged to near-optimal $Q$-graphs that provide an almost tight lower bound. The difference between the upper and lower bounds is negligible, with an order of at most $\sim 10^{-5}$.

\begin{figure}[t]
\centering
    \includegraphics[scale = 0.35]{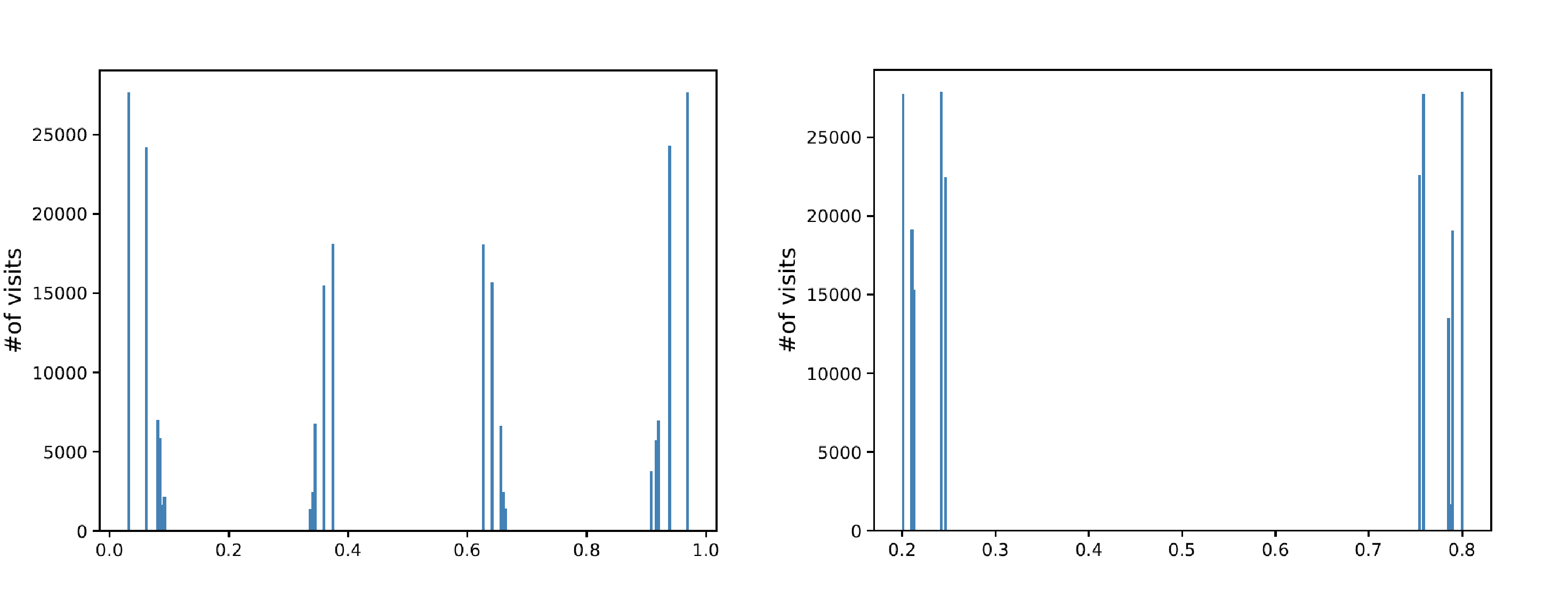}
    \caption{Histograms of the MDP states visited under an estimated optimal policy learned by RL in the case of the N-Ising channel. The left figure presents the histogram of MDP states when the channel state parameter is $\epsilon=0.1$, while the right figure presents the histogram when the channel state parameter is $\epsilon=0.4$.}
    \label{fig:N-Ising_Qgraph}
\end{figure}

\section{Conclusions}\label{sec:conclusion}
In this paper, we presented a novel approach for deriving upper bounds on the feedback capacity of FSCs. Our upper bounds leverage an extension of the duality upper bound from mutual information to the case of directed information. The bounds are expressed as a function of a test distribution, and we propose the use of graph-based test distributions as an extension to the commonly used Markov test distributions. For a fixed graph-based test distribution, we demonstrate that the upper bounds can be formulated as a MDP problem, making them computationally tractable. Moreover, in the cases of unifilar FSCs and finite-memory state channels, the MDP formulation exhibits finite states, actions, and disturbances. As a result, for such channels, our methodology is capable of handling channels with large alphabets, and deriving analytical upper bounds becomes relatively straightforward through the solution of the associated Bellman equation. Additionally, we have introduced several exploration methodologies of $Q$-graphs and evaluated their effectiveness on well-known channels, yielding either tight or near-optimal bounds on their capacity.

\begin{appendices}
\section{Derivation of the upper bound --- Proof of Theorem \ref{Th: DUB_Q}} \label{app: DUB_Q}
In this section, we prove Theorem \ref{Th: DUB_Q}. We simplify \eqref{eq: DB_DI} for the case where the test distribution is defined on a $Q$-graph. For convenience, we omit the dependence on the initial pair $(s_0,q_0)$. For a fixed $(s_0,q_0)$, consider the following chain of inequalities:
\begin{align}\label{eq: DUB_Q}
	&D(P_{Y^n\|X^n=x^n}\|T_{Y^n}) \nn\\&= \sum_{y^n} P(y^n\|x^n)\log_2\left(\frac{P(y^n\|x^n)}{T_{Y^n}(y^n)}\right)\nn
	\\&=  \sum_{y^n} P(y^n\|x^n)\log_2\left(\prod_{i=1}^n\frac{P(y_i|y^{i-1},x^i)}{T(y_i|y^{i-1})}\right)\nn
	\\&\stackrel{(a)}= \sum_{i=1}^n \sum_{y^n} P(y^n\|x^n)\log_2\left(\frac{P(y_i|y^{i-1},x^i)}{T(y_i|y^{i-1})}\right)\nn
	\\&\stackrel{(b)}= \sum_{i=1}^n \sum_{y^i}\sum_{q_{i-1}} P(y^i\|x^i)\log_2\left(\frac{P(y_i|y^{i-1},x^i)}{T_{Y|Q}(y_i|q_{i-1})}\right)\cdot\mathbbm{1}{\{q_{i-1}=\Phi(y^{i-1})\}}\nn
    \\&= \sum_{i=1}^n \sum_{y^{i-1}}\sum_{q_{i-1}} P(y^{i-1}\|x^{i-1})\mathbbm{1}{\{q_{i-1}=\Phi(y^{i-1})\}}\sum_{y_i}P(y_i|y^{i-1},x^i)\log_2\left(\frac{P(y_i|y^{i-1},x^i)}{T_{Y|Q}(y_i|q_{i-1})}\right)\nn
	\\&\stackrel{(c)}= \sum_{i=1}^n \sum_{y^{i-1}}\sum_{q_{i-1}} P(y^{i-1}\|x^{i-1})\mathbbm{1}{\{q_{i-1}=\Phi(y^{i-1})\}}\cdot D\left(P_{Y_i|Y^{i-1},X^i}(\cdot|y^{i-1},x^i)\;\|\;T(\cdot|q_{i-1})\right)\nn	
	\\&\le \max_{f(x^n\|y^{n-1})}\sum_{i=1}^n\mathbb{E}\left[ D\left(P(\cdot|Y^{i-1},x^i)\;\|\;T\left(\cdot|Q_{i-1}\right)\right)\right]\nn
	\\&\stackrel{(d)}=\max_{f(x^n\|y^{n-1})}\sum_{i=1}^n\mathbb{E}\left[ D\left(\sum_{s_{i-1}}P(s_{i-1}|Y^{i-1},x^{i-1})\cdot P_{Y|X,S}(\cdot|x_i,s_{i-1})\;\Big{\|}\;T\left(\cdot|Q_{i-1}\right)\right)\right],
\end{align}
where $(a)$ follows by exchanging the order of summation, $(b)$ follows by  marginalizing over $y_{i+1}^n$ and the fact that $q_{i-1}$ is a function of $y^{i-1}$, $(c)$ follows by identifying the relative entropy, and $(d)$ follows by the Markov chain $S_{i-1}-(X^{i-1},Y^{i-1},S_0)-X_i$ and the channel law.

By dividing the term in \eqref{eq: DUB_Q} by $n$, minimizing over $(s_0,q_0)$, and taking the limit we obtain
\begin{align}
    \mathsf{C_{fb}} &\le \lim_{n\to\infty}\min_{s_0,q_0}\max_{f(x^n\|y^{n-1})}\frac{1}{n}\sum_{i=1}^n\mathbb{E}\left[D\left(\sum_{s_{i-1}}P(s_{i-1}|Y^{i-1},x^{i-1},s_0)P(\cdot|x_i,s_{i-1})\Big{\|}T\left(\cdot|Q_{i-1}\right)\right)\right]\label{eq: main_ub}\\&= \lim_{n\to\infty}\max_{f(x^n\|y^{n-1},s_0)}\min_{s_0,q_0}\frac{1}{n}\sum_{i=1}^n\mathbb{E}\left[D\left(\sum_{s_{i-1}}P(s_{i-1}|Y^{i-1},x^{i-1},s_0)P(\cdot|x_i,s_{i-1})\Big{\|}T\left(\cdot|Q_{i-1}\right)\right)\right],\label{eq: main_ub2}
\end{align}
where the existence of the limit is shown next.

\subsection{Existence of the limit in \eqref{eq: main_ub}}\label{app:subsec_limexists}
In this section, we prove the existence of the limit in \eqref{eq: main_ub}, consequently implying the existence of the limit in \eqref{eq: main_ub2}. We first introduce a technical lemma that will be utilized in the proof. The proof of the lemma is provided in Appendix \ref{app: lemmas}.
\begin{lemma}\label{lemma: conditional_diff}
    For any FSC $P_{S^+,Y|X,S}$ and a test distribution $T_{Y^n}$, and for any $x^n,y^m,s_0$ s.t. $n> m\ge 1$, the following bound holds:
    \begin{align}\label{eq: diff_kl}
     \Big{|}\infdiv[\Big]{P_{Y_{m+1}^n\|x_{m+1}^n|x^m,y^m,S_m,s_0}}{T_{Y_{m+1}^n|y^m,q_0}\Big{|}P_{S_m|x^{m+1},y^m,s_0}}&\nn\\-\infdiv[\Big]{P_{Y_{m+1}^n\|x_{m+1}^n|x^m,y^m,s_0}}{T_{Y_{m+1}^n|y^m,q_0}}&\Big{|}\leq
    \log_2(|\mathcal{S}|).
    \end{align}
    That is, the divergence is changed by at most $\log_2(|\mathcal{S}|)$ when conditioning over $S_m$.
\end{lemma}

We proceed to show the existence of the limit. Let us define
\begin{align}\label{c_quant}
\mathsf{c}(x^n,s_0,q_0) \triangleq \frac{1}{n} D\left(P_{Y^n\|X^n=x^n,S_0=s_0}\|T_{Y^n|Q_0=q_0}\right).
\end{align}
In addition, we define $\underline{\mathsf{C}}_n$ as
\begin{align}\label{eq: c__under_limit}
\underline{\mathsf{C}}_n &\triangleq \min_{s_0,q_0}\max_{f(x^n\|y^{n-1})}\mathsf{c}(x^n,s_0,q_0) \nn\\ &=\frac{1}{n}\min_{s_0,q_0}\max_{f(x^n\|y^{n-1})} D\left(P_{Y^n\|X^n=x^n,S_0=s_0}\|T_{Y^n|Q_0=q_0}\right).
\end{align}
Our objective is to show the existence of the limit $\lim_{n\to\infty} \underline{\mathsf{C}}_n$, which in turn implies the existence of the limit in \eqref{eq: main_ub}. The idea of the proof is to show that the sequence $n\left(\underline{\mathsf{C}}_n-\frac{\log_2(|\mathcal{S}|)}{n}\right)$ is super-additive. Then, by Fekete's lemma \cite{Fekete1923}, the limit $\lim\limits_{n\to\infty} \underline{\mathsf{C}}_n$ exists, and is equal to $\sup_n \left[\underline{\mathsf{C}}_n-\frac{\log_2|\mathcal{S}|}{n}\right]$. Recall that a sequence $a_n$ is super-additive if it satisfies the inequality $a_{m+k}\geq a_m+a_k$ for any positive integers $m$ and $k$.

Consider two positive integers $m$ and $k$ such that $m+k = n$. Assume that $\hat{f}_1(x^m\|y^{m-1})$ and $\hat{f}_2(x^k\|y^{k-1})$ are the mapping functions that achieve the maximum for $\underline{\mathsf{C}}_m$ and $\underline{\mathsf{C}}_k$ in \eqref{eq: c__under_limit}, respectively. Define a new mapping $\hat{f}(x^n\|y^{n-1})$ as follows: 
\begin{align}
    \hat{f}(x^n\|y^{n-1}) = \hat{f}_1(x^m\|y^{m-1})\hat{f}_2(x_{m+1}^n\|y_{m+1}^{n-1}).
\end{align}
Accordingly, since $\hat{f}(x^n\|y^{n-1})$ is not necessarily the input mapping that achieves the maximum for $\underline{\mathsf{C}}_n$, then under the choice of $\hat{f}(x^n\|y^{n-1})$ we have
\begin{align} \label{eq: limit_step1}
    n\underline{\mathsf{C}}_n &\geq \min_{s_0,q_0} D\left(P_{Y^n\|x^n,s_0}\|T_{Y^n|q_0}\right) \nn \\
    &\stackrel{(a)}= \min_{s_0,q_0}\bigg[\infdiv[\big]{P_{Y^m\|x^m,s_0}}{T_{Y^m|q_0}} + \infdiv[\Big]{P_{Y_{m+1}^n\|x_{m+1}^n|Y^m,x^m,s_0}}{T_{Y_{m+1}^n|Y^m,q_0}\Big{|}P_{Y^m\|x^m,s_0}}\bigg]\nn\\
    &\stackrel{(b)}\geq \min_{s_0,q_0}\infdiv[\big]{P_{Y^m\|x^m,s_0}}{T_{Y^m|q_0}} + \min_{s_0,q_0}\infdiv[\Big]{P_{Y_{m+1}^n\|x_{m+1}^n|Y^m,x^m,s_0}}{T_{Y_{m+1}^n|Y^m,q_0}\Big{|}P_{Y^m\|x^m,s_0}}\nn\\
    &= m\underline{\mathsf{C}}_m + \min_{s_0,q_0}\infdiv[\Big]{P_{Y_{m+1}^n\|x_{m+1}^n|Y^m,x^m,s_0}}{T_{Y_{m+1}^n|Y^m,q_0}\Big{|}P_{Y^m\|x^m,s_0}},
\end{align}
where $(a)$ follows by the chain rule for relative entropy, $(b)$ follows from $\min_t\left[f(t)+g(t)\right]\geq \min_t f(t)+\min_t g(t)$.

We show now that the second term in \eqref{eq: limit_step1} is at least $k\underline{\mathsf{C}}_k$. For any initial pair $(s_0,q_0)$ we have
\begin{align} \label{eq: limit_step2}
    &\infdiv[\Big]{P_{Y_{m+1}^n\|x_{m+1}^n|Y^m,x^m,s_0}}{T_{Y_{m+1}^n|Y^m,q_0}\Big{|}P_{Y^m\|x^m,s_0}}\nn\\
    &\stackrel{(a)}\geq \infdiv[\Big]{P_{Y_{m+1}^n\|x_{m+1}^n|Y^m,x^m,S_m,s_0}}{T_{Y_{m+1}^n|Y^m,q_0}\Big{|}P_{Y^m\|x^m,s_0}P_{S_m|Y^m,x^{m+1},s_0}} - \log_2(|\mathcal{S}|)\nn\\
    &\stackrel{(b)}= \infdiv[\Big]{P_{Y_{m+1}^n\|x_{m+1}^n|Y^m,x^m,S_m,s_0}}{T_{Y_{m+1}^n|Q_m,Y^m,q_0}\Big{|}P_{Y^m\|x^m,s_0}P_{S_m,Q_m|Y^m,x^{m+1},s_0,q_0}} - \log_2(|\mathcal{S}|)\nn\\
    &\stackrel{(c)}= \infdiv[\Big]{P_{Y_{m+1}^n\|x_{m+1}^n,S_m}}{T_{Y_{m+1}^n|Q_m}\Big{|}P_{Y^m\|x^m,s_0}P_{S_m,Q_m|Y^m,x^{m+1},s_0,q_0}} - \log_2(|\mathcal{S}|)\nn\\
    &= \sum_{y^m}P(y^m\|x^m,s_0)\sum_{s_m,q_m}P(s_m,q_m|y^m,x^{m+1},s_0,q_0)\infdiv[\Big]{P_{Y_{m+1}^n\|x_{m+1}^n,s_m}}{T_{Y_{m+1}^n|q_m}} - \log_2(|\mathcal{S}|)\nn\\
    &\ge\sum_{y^m}P(y^m\|x^m,s_0)\min_{s_m,q_m}\infdiv[\Big]{P_{Y_{m+1}^n\|x_{m+1}^n,s_m}}{T_{Y_{m+1}^n|q_m}} - \log_2(|\mathcal{S}|)\nn\\
    &= \min_{s_m,q_m}\infdiv[\Big]{P_{Y_{m+1}^n\|x_{m+1}^n,s_m}}{T_{Y_{m+1}^n|q_m}} - \log_2(|\mathcal{S}|)\nn\\
    &=  k\underline{\mathsf{C}}_k - \log_2(|\mathcal{S}|).
\end{align}
where $(a)$ follows since the relative entropy term is changed by at most $\log_2(|\mathcal{S}|)$ when conditioning on $S_m$ (see Lemma \ref{lemma: conditional_diff}), $(b)$ follows because $Q_m=\Phi(Y^m)$, and $(c)$ follows due to the Markov chain $Y_{m+1}^n-(S_m,X_{m+1}^n)-(X^m,Y^m,S_0)$ and since $T(y_{m+1}^n|q_m,y^m,q_0)=T(y_{m+1}^n|q_m)$.

Hence, from \eqref{eq: limit_step1} and \eqref{eq: limit_step2}, we observe that
\begin{align*}
    n\underline{\mathsf{C}}_n \geq m\underline{\mathsf{C}}_m + k\underline{\mathsf{C}}_k - \log_2(|\mathcal{S}|).
\end{align*}
By rearranging the inequality, we obtain
\begin{align*}
    n\left[\underline{\mathsf{C}}_n-\frac{\log_2(|\mathcal{S}|)}{n}\right]\ge m\left[\underline{\mathsf{C}}_m-\frac{\log_2(|\mathcal{S}|)}{m}\right] + k\left[\underline{\mathsf{C}}_k-\frac{\log_2(|\mathcal{S}|)}{k}\right].
\end{align*}
Accordingly, we obtained that the sequence $n\left(\underline{\mathsf{C}}_n-\frac{\log_2(|\mathcal{S}|)}{n}\right)$ is super-additive which concludes the existence of the limit $\lim\limits_{n\to\infty} \underline{\mathsf{C}}_n$, as required.

\subsection{Proof of Lemma \ref{lemma: conditional_diff}}\label{app: lemmas}
In the following, we introduce a new random vector $\tilde{S}^n$ defined as: $\tilde{S}_{m+1}=S_m$ and $\tilde{S}_i=0$ for $i\neq m+1$. We begin by establishing a preliminary lemma that will be used for the proof of Lemma \ref{lemma: conditional_diff}.
\begin{lemma}\label{lemma: marg}
For any FSC, and any $x^n,y^n,s_0$, and $n>m\ge 1$, the following equality holds:
\begin{align}\label{eq: marg}
    \sum_{s_m}P(s_m|x^{m+1},y^m,s_0)P(y_{m+1}^n\|x_{m+1}^n|y^m,x^m,s_m,s_0) = P(y_{m+1}^n\|x_{m+1}^n|x^m,y^m,s_0).
    \end{align}
\end{lemma}
\begin{proof}[Proof of Lemma \ref{lemma: marg}]
    We begin with the following derivation:
    \begin{align}\label{eq: marg_step1}
        P(x_{m+1}^n,y_{m+1}^n|x^m,y^m,s_0)&=\sum_{\tilde{s}_{m+1}^n}P(x_{m+1}^n,y_{m+1}^n,\tilde{s}_{m+1}^n|x^m,y^m,s_0)\nn\\
        &=\sum_{\tilde{s}_{m+1}^n}P(x_{m+1}^n\|y_{m+1}^{n-1},\tilde{s}_{m+1}^{n-1}|x^m,y^m,s_0)P(y_{m+1}^n,\tilde{s}_{m+1}^n\|x_{m+1}^n|x^m,y^m,s_0)\nn\\
        &\stackrel{(a)}=\sum_{\tilde{s}_{m+1}^n}P(x_{m+1}^n\|y_{m+1}^{n-1}|x^m,y^m,s_0)P(y_{m+1}^n,\tilde{s}_{m+1}^n\|x_{m+1}^n|x^m,y^m,s_0)\nn\\
        &=P(x_{m+1}^n\|y_{m+1}^{n-1}|x^m,y^m,s_0)\sum_{\tilde{s}_{m+1}^n}P(y_{m+1}^n,\tilde{s}_{m+1}^n\|x_{m+1}^n|x^m,y^m,s_0),
    \end{align}
    where $(a)$ follows from the Markov chain $X_i-(X^{i-1},Y^{i-1},S_0)-S_{m}$, which holds for any $i\geq m+1$. On the other hand, we have
    \begin{align}\label{eq: marg_step2}
        P(x_{m+1}^n,y_{m+1}^n|x^m,y^m,s_0) = P(x_{m+1}^n\|y_{m+1}^{n-1}|x^m,y^m,s_0)P(y_{m+1}^n\|x_{m+1}^n|x^m,y^m,s_0).
    \end{align}
By comparing \eqref{eq: marg_step1} and \eqref{eq: marg_step2}, we deduce that 
\begin{align}
\sum_{\tilde{s}_{m+1}^n}P(y_{m+1}^n,\tilde{s}_{m+1}^n\|x_{m+1}^n|x^m,y^m,s_0)=P(y_{m+1}^n\|x_{m+1}^n|x^m,y^m,s_0).
\end{align}
The proof is thereby concluded, noting that the left-hand side of Eq. \eqref{eq: marg} is equal to $\sum_{\tilde{s}_{m+1}^n}P(y_{m+1}^n,\tilde{s}_{m+1}^n\|x_{m+1}^n|x^m,y^m,s_0)$ due to the unique construction of $\tilde{s}^n$.
\end{proof}
\begin{proof}[Proof of Lemma \ref{lemma: conditional_diff}]
We proceed with the proof of Lemma \ref{lemma: conditional_diff}. For convenience, we omit the dependence on the initial pair $(s_0,q_0)$, and bound the difference in \eqref{eq: diff_kl} by $\log_2(\mathcal{|S|})$ as follows:
\begin{align}
    &\Big{|}\infdiv[\Big]{P_{Y_{m+1}^n\|x_{m+1}^n|x^m,y^m,S_m}}{T_{Y_{m+1}^n|y^m}\Big{|}P_{S_m|x^{m+1},y^m}}-\infdiv[\Big]{P_{Y_{m+1}^n\|x_{m+1}^n|x^m,y^m}}{T_{Y_{m+1}^n|y^m}}\Big{|}\nn\\
    &\stackrel{(a)}\leq \Bigg{|}\sum_{s_m}P(s_m|x^{m+1},y^m)\sum_{y_{m+1}^n}P(y_{m+1}^n\|x_{m+1}^n|x^m,y^m,s_m)\log_2(P(y_{m+1}^n\|x_{m+1}^n|x^m,y^m,s_m))
    \nn\\&\quad-\sum_{y_{m+1}^n}P(y_{m+1}^n\|x_{m+1}^n|x^m,y^m)\log_2(P(y_{m+1}^n\|x_{m+1}^n|x^m,y^m))\Bigg{|}\nn\\
    &\quad +\Bigg{|}\sum_{s_m}P(s_m|x^{m+1},y^m)\sum_{y_{m+1}^n}P(y_{m+1}^n\|x_{m+1}^n|x^m,y^m,s_m)\log_2(T(y_{m+1}^n|y^m))
    \nn\\&\quad-\sum_{y_{m+1}^n}P(y_{m+1}^n\|x_{m+1}^n|x^m,y^m)\log_2(T(y_{m+1}^n|y^m))\Bigg{|}\nn\\
    &\stackrel{(b)}=\Bigg{|}\sum_{s_m}P(s_m|x^{m+1},y^m)\sum_{y_{m+1}^n}P(y_{m+1}^n\|x_{m+1}^n|x^m,y^m,s_m)\log_2(P(y_{m+1}^n\|x_{m+1}^n|x^m,y^m,s_m))
    \nn\\&\quad-\sum_{y_{m+1}^n}P(y_{m+1}^n\|x_{m+1}^n|x^m,y^m)\log_2(P(y_{m+1}^n\|x_{m+1}^n|x^m,y^m))\Bigg{|}\nn\\
    &\stackrel{(c)}= \Bigg{|}\sum_{\tilde{s}_{m+1}^n,y_{m+1}^n}    P(\tilde{s}_{m+1}^n,y_{m+1}^n\|x_{m+1}^n|x^m,y^m)\log_2(P(y_{m+1}^n\|x_{m+1}^n|x^m,y^m,\tilde{s}_{m+1}^n))
    \nn\\&\quad
    -\sum_{y_{m+1}^n}P(y_{m+1}^n\|x_{m+1}^n|x^m,y^m)\log_2(P(y_{m+1}^n\|x_{m+1}^n|x^m,y^m))\Bigg{|}\nn\\
    &\triangleq \Big{|}H(Y_{m+1}^n\|x_{m+1}^n|x^m,y^m) - H(Y_{m+1}^n\|x_{m+1}^n|\tilde{S}_{m+1}^n,x^m,y^m)\Big{|}\nn\\    
    &\triangleq I(\tilde{S}_{m+1}^n; Y_{m+1}^n\|x_{m+1}^n|x^m,y^m)\nn\\
    &\leq H(\tilde{S}_{m+1}^n\|x_{m+1}^n|x^m,y^m)
    \nn\\
    &\leq
    \log_2(|\mathcal{S}|), 
\end{align}
where $(a)$ follows by explicitly extracting the relative entropies and then applying the triangle inequality, $(b)$ follows since $\sum_{s_m}P(s_m|x^{m+1},y^m,s_0)P(y_{m+1}^n\|x_{m+1}^n|x^m,y^m,s_m,s_0)=P(y_{m+1}^n\|x_{m+1}^n|x^m,y^m,s_0)$ (see Lemma \ref{lemma: marg}), which implies that the second absolute value expression is equal to zero, and $(c)$ follows directly by the fact that $\tilde{S}_{m+1}^n\triangleq\{S_m,0,\dots,0\}$.
\end{proof}

\section{MDP Formulation of the Dual Upper Bound (Theorem \ref{theorem: formulation})} \label{app: formulation}
In this section, we prove Theorem \ref{theorem: formulation}, which concerns the representation of the upper bound in Theorem \ref{Th: DUB_Q} as an MDP. The proof consists of four technical parts, summarized in Lemma \ref{lemma: DP1}. The last part of the lemma relates the average reward induced by the MDP to the upper bound in Theorem \ref{Th: DUB_Q}.

\begin{lemma}\label{lemma: DP1}
\begin{enumerate}
    \item The reward is a time-invariant function of $z_{t-1}$ and $x_t$.
    \item The MDP state, $z_t$, is a time-invariant function of $z_{t-1}$, $y_t$, and $x_t$.
    \item Given $z_{t-1}$ and $x_t$, the disturbance, $y_t$, is conditionally independent of the past.
    \item The limit and the maximization in the upper bound in Theorem \ref{Th: DUB_Q} can be exchanged, i.e.,
    \begin{align} 
        &\lim_{n\to\infty}\min_{s_0,q_0}\max_{f(x^n\|y^{n-1})} c(x^n,s_0,q_0)=\sup\liminf_{n\to\infty}\min_{s_0,q_0}c(x^n,s_0,q_0),\nn
    \end{align}
    where the supremum is over all deterministic functions $\{f_i:\mathcal{X}^{i-1}\times\mathcal{Y}^{i-1}\rightarrow \mathcal{X}\}_{i\ge 1}$, and $c(x^n,s_0,q_0)$ is defined in \eqref{c_quant}.
\end{enumerate}
Accordingly, we can conclude from the last item that $\mathsf{C_{fb}}\le \rho^*$.

\end{lemma}
\begin{proof}[Proof of Lemma \ref{lemma: DP1}]
\begin{enumerate}
    \item The reward function in \eqref{reward_1} is defined as
    \begin{align*}
	    g(z_{t-1},x_t) \triangleq D\left(\sum_{s_{t-1}}\beta_{t-1}(s_{t-1}) P_{Y|X,S}(\cdot|x_t,s_{t-1})\middle\|T_{Y|Q}(\cdot|q_{t-1})\right).
    \end{align*}
    Clearly, for a fixed FSC and test distribution, the reward is a time-invariant function of the previous MDP state $z_{t-1}=(\beta_{t-1},q_{t-1})$ and the action $x_t$. 
    \item Let us first express $\beta_t$ as a function of $(\beta_{t-1},x_t,y_t)$. For any $s_t\in\mathcal{S}$, we have
    \begin{align} \label{Next_FSC}
	    \beta_{t}(s_t) &= P(s_t|x^t,y^t,s_0) \nn
        	\\&= \frac{P(y_t,s_t|x^{t},y^{t-1},s_0)}{P(y_t|x^{t},y^{t-1},s_0)} \nn
			\\&= \frac{\sum_{s_{t-1}}P(s_{t-1},y_t,s_t|x^{t},y^{t-1},s_0)}{\sum_{\tilde{s}_{t-1}}P(\tilde{s}_{t-1},y_t|x^{t},y^{t-1},s_0)} \nn
            \\&\stackrel{(a)}= \frac{\sum_{s_{t-1}}P(s_{t-1}|x^{t-1},y^{t-1},s_0)P(y_t|s_{t-1},x_t)P(s_t|x_t,y_t,s_{t-1})}{\sum_{\tilde{s}_{t-1}}P(\tilde{s}_{t-1}|x^{t-1},y^{t-1},s_0)P(y_t|\tilde{s}_{t-1},x_t)} \nn
            \\&= \frac{\sum_{s_{t-1}}\beta_{t-1}(s_{t-1})P(y_t|s_{t-1},x_t)P(s_t|x_t,y_t,s_{t-1})}{\sum_{\tilde{s}_{t-1}}\beta_{t-1}(\tilde{s}_{t-1})P(y_t|\tilde{s}_{t-1},x_t)},
    \end{align}
    where $(a)$ follows from the Markov chain $S_{t-1}-(X^{t-1},Y^{t-1},S_0)-X_t$ and the channel law.
    Accordingly, $\beta_t$ is a deterministic function of the previous MDP state, $z_{t-1}$, the action $x_t$, and the disturbance $y_t$. Further, since $q_t = \phi(q_{t-1},y_t)$, then there exists a time-invariant function $F(\cdot)$ such that $z_t = F(z_{t-1},x_t,y_t)$.
    
    \item In this item, we show that $P(y_t|z^{t-1},y^{t-1},x^t) = P(y_t|z_{t-1},x_t)$. In particular, 
    \begin{align}
        P(y_t|z^{t-1},y^{t-1},x^t) 
        &= \sum_{s_{t-1}}P(s_{t-1},y_t|z^{t-1},y^{t-1},x^t) \nn\\
        &\stackrel{(a)}=\sum_{s_{t-1}}P(s_{t-1}|z^{t-1},y^{t-1},x^t)P(y_t|x_t,s_{t-1},z_{t-1}) \nn\\ 
        &\stackrel{(b)}=\sum_{s_{t-1}}P(s_{t-1}|z_{t-1},x_t)P(y_t|x_t,s_{t-1},z_{t-1}) \nn\\ 
        &= \sum_{s_{t-1}}P(s_{t-1},y_t|z_{t-1},x_t) \nn\\
        &= P(y_t|z_{t-1},x_t),
    \end{align}
    where $(a)$ follows from applying the chain rule and the Markov property induced by the channel law, and $(b)$ follows since $z_{t-1}$ consists of $\beta_{t-1} = P(S_{t-1}|x^{t-1}, y^{t-1})$.
    
    \item The proof here is similar to the proof of Lemma $6$ in \cite{Huleihel_Sabag_DB}, which is based on the super-additivity property of the sequence $n\left(\underline{\mathsf{C}}_n-\frac{\log_2(|\mathcal{S}|)}{n}\right)$. However, for completeness, as it is not a special case of this lemma, we provide here the proof. The equality will be shown by proving the corresponding two inequalities. The first inequality can be shown as follows:
    \begin{align} 
        \lim_{n\to\infty}\underline{\mathsf{C}}_n 
        &= \lim_{n\to\infty}\min_{s_0,q_0}\max_{f(x^n\|y^{n-1})}c(x^n,s_0,q_0)\nn\\
        &= \lim_{n\to\infty}\max_{f(x^n\|y^{n-1},s_0)}\min_{s_0,q_0}c(x^n,s_0,q_0)\nn\\
        &\stackrel{(a)}=\sup_n\max_{f(x^n\|y^{n-1},s_0)}\min_{s_0,q_0}\left[c(x^n,s_0,q_0)-\frac{\log_2|\mathcal{S}|}{n}\right]\nn\\
        &=\sup\sup_n\min_{s_0,q_0}\left[c(x^n,s_0,q_0)-\frac{\log_2|\mathcal{S}|}{n}\right]\nn\\
        &\ge\sup\liminf_{n\to\infty}\min_{s_0,q_0}\left[c(x^n,s_0,q_0)-\frac{\log_2|\mathcal{S}|}{n}\right]\nn\\
        &=\sup\liminf_{n\to\infty}\min_{s_0,q_0}c(x^n,s_0,q_0),
    \end{align}
    where $(a)$ follows by Fekete's lemma, and the supremum in the last three steps is over all deterministic functions $\{f_i:\mathcal{X}^{i-1}\times\mathcal{Y}^{i-1}\times\mathcal{S}\rightarrow \mathcal{X}\}_{i\ge 1}$.

    We now show the reverse inequality. Using the notation and the main result from Appendix~\ref{app:subsec_limexists}, the existence of $\lim\limits_{n\to\infty} \underline{\mathsf{C}}_n$ implies that, for any $\epsilon>0$, there exists an $N(\epsilon)$ such that for all $k>N(\epsilon)$
    \begin{align} \label{eq: n_eps}
        \underline{\mathsf{C}}_k \geq \lim_{n\to\infty} \underline{\mathsf{C}}_n -\epsilon.
    \end{align}

    Fix $k>N(\epsilon)$, and let $\hat{f}(x^k\|y^{k-1})$ be the input mapping that achieves the maximum for $\underline{\mathsf{C}}_k$. Let us construct 
    \begin{align}
        \tilde{f}(x^n\|y^{n-1}) = \hat{f}(x^k\|y^{k-1})\hat{f}(x_{k+1}^{2k}\|y_{k+1}^{2k-1})\hat{f}(x_{2k+1}^{3k}\|y_{2k+1}^{3k-1})\cdots.
    \end{align}
    Consider now the following chain of inequalities
    \begin{align}
    &\sup\liminf_{n\to\infty}\min_{s_0,q_0}c(x^n,s_0,q_0)\nn\\ 
    &= \sup\liminf_{n\to\infty}\min_{s_0,q_0}\frac{1}{n} D\left(P_{Y^n\|X^n=x^n,S_0=s_0}\|T_{Y^n|Q_0=q_0}\right)\nn\\
    &\stackrel{(a)}\ge \liminf_{n\to\infty}\min_{s_0,q_0}\frac{1}{n} D\left(P_{Y^n\|X^n=\tilde{x}^n,S_0=s_0}\|T_{Y^n|Q_0=q_0}\right)\nn\\
    &\stackrel{(b)}= \liminf_{n\to\infty}\min_{s_0,q_0}\frac{1}{n} \sum_{i=0}^{\left\lfloor \frac{n}{k} \right\rfloor - 1}D\left(P_{Y_{ik+1}^{(i+1)k}\|\tilde{x}_{ik+1}^{(i+1)k}|Y^{ik},\tilde{x}^{ik},s_0}\Big{\|}T_{Y_{ik+1}^{(i+1)k}|Y^{ik},q_0}\Big{|}P_{Y^{ik}\|\tilde{x}^{ik},s_0}\right)\nn\\
    &\stackrel{(c)}\ge \liminf_{n\to\infty}\frac{1}{n} \sum_{i=0}^{\left\lfloor \frac{n}{k} \right\rfloor - 1}\min_{s_{ik},q_{ik}}D\left(P_{Y_{ik+1}^{(i+1)k}\|\tilde{x}_{ik+1}^{(i+1)k},s_{ik}}\Big{\|}T_{Y_{ik+1}^{(i+1)k}|q_{ik}}\right)\nn\\
    &\stackrel{(d)}\ge \liminf_{n\to\infty}\min_{s_0,q_0}\frac{k}{n}\left\lfloor\frac{n}{k}\right\rfloor\left[\frac{1}{k}\cdot D\left(P_{Y^{k}\|\tilde{x}^k,s_0}\Big{\|}T_{Y^{k}\|q_0}\right)\right]\nn\\    
    &\stackrel{(e)}\ge 
        \lim_{n\to\infty}\underline{\mathsf{C}}_n-\epsilon,
    \end{align}
    where $(a)$ follows by considering the inputs that follow the mapping $\tilde{f}(x^n\|y^{n-1})$, which is not necessarily the mapping that achieves the maximum, $(b)$ follows from the chain rule and the fact that $k$ is fixed and the divergence is bounded (otherwise the bound is meaningless), and therefore, when rounding $n$ to $k\lfloor n/k\rfloor$ the residual goes to zero, $(c)$ holds by following the derivation in \eqref{eq: limit_step2} for each relative entropy term in the sum, $(d)$ follows since each term in the sum is identical due to the construction of the input distribution, and, finally, $(e)$ follows from \eqref{eq: n_eps}.
\end{enumerate}  
\end{proof}

\section{NOST Channel --- Lower Bound} \label{app: NOST_LB}
\begin{proof}
In this section, we derive a lower bound on the capacity of the NOST channel to demonstrate the tightness of our upper bound and conclude the proof of Theorem \ref{th: NOST}. We establish the lower bound by employing the $Q$-graph lower bound derived in \cite{OronBasharfeedback}, which holds for unifilar FSCs. To introduce this lower bound, we first define a property termed \textit{BCJR-invariant} input distribution. An input distribution is considered BCJR-invariant if it satisfies the Markov chain $S^+-Q^+-(Q,Y)$. A simple verification of this Markov chain is given by the following equation:
\begin{align}\label{eq: BCJR}
    \pi_{S|Q}(s^+|q^+) = \frac{\sum_{x,s}\pi_{S|Q}(s|q)P(x|s,q)P(y|x,s)\mathbbm{1}_{\{s^+=\tilde{f}(x,y,s)\}}}{\sum_{x',s'}\pi_{S|Q}(s'|q)P(x'|s',q)P(y|x',s')},
\end{align}
which needs to hold for all $(s^+, q, y)$ and $q^+ = \phi(q,y)$, where $\pi_{S,Q} = \pi_{Q}\pi_{S|Q}$ is the induced stationary distribution on the $(S,Q)$-graph\footnote{The $(S,Q)$-graph is a directed and connected graph that integrates both the $Q$-graph and the evolution of the channel states. For further details, refer to \cite{OronBasharfeedback}.}.

Having defined a BCJR-invariant input distribution, we proceed to introduce the $Q$-graph lower bound from \cite{OronBasharfeedback} through the following theorem.
 \begin{theorem}\cite[Theorem $3$]{Sabag_UB_IT}\label{theorem:bcjr_general}
If the initial state $s_0$ is available to both the encoder and the decoder, then the feedback capacity of a strongly connected unifilar FSC is bounded by
\begin{align}\label{eq:Theorem_Lower}
\mathsf{C_{fb}}&\geq I(X,S;Y|Q),
\end{align}
for all aperiodic inputs $P_{X|S,Q}$ that are BCJR-invariant, and for all $Q$-graphs with $q_0$ such that $(s_0,q_0)$ lies in an aperiodic closed communicating class.
\end{theorem}
Hereafter, we refer to a pair comprising a $Q$-graph and an input distribution $P_{X|S,Q}$ that satisfies the BCJR-invariant property as a \emph{graph-based encoder}.

It is essential to note that the NOST channel, being a finite-memory state channel, does not meet the unifilar channel requirement for the lower bound. Consequently, we reformulate the channel to a unifilar FSC following Section \ref{subsection: FC_extention}. Once reformulated, the transformed channel is defined with the following transition probability
\begin{align}\label{eq: trapdoor_input_dist}
    P_{Y|X,\tilde{S}}(0|x,\tilde{s}) = \scalebox{0.93}{\begin{tabularx}{0.35\textwidth} { 
  |@{}p{1.3cm}@{} 
  | >{\centering\arraybackslash}X  
  | >{\centering\arraybackslash}X | }
 \hline
          & $x=0$ & $x=1$\\
 \hline
$\;\;\tilde{s}=0$  & $1-0.5\epsilon$  & $0.5(1-\epsilon)$  \\
\hline
$\;\;\tilde{s}=1$  &  $0.5(1+\epsilon)$   & $0.5\epsilon$   \\
\hline
\end{tabularx}},
\end{align}
while the new channel state evolves according to $\tilde{s}_t=y_t$.

Using the transformed channel, we derive the lower bound employing a specific graph-based encoder. We demonstrate that the achievable rate resulting from this graph-based encoder precisely matches the upper bound derived in the preceding section. Also here, the $Q$-graph used in deriving the lower bound is the first-order Markov $Q$-graph in Fig. \ref{fig:1Markov}. For this $Q$-graph, we define the following input distribution, where $a$ is defined in Eq. \eqref{eq: opt_a}:
\begin{align}\label{eq: trapdoor_input_dist}
    P_{X|\tilde{S},Q}(x|\tilde{s},q) = \scalebox{0.93}{\begin{tabularx}{0.35\textwidth} { 
  |@{}p{1.3cm}@{} 
  | >{\centering\arraybackslash}X  
  | >{\centering\arraybackslash}X | }
 \hline
          & $x=0$ & $x=1$\\
 \hline
$\;\;q=1$  & $2a-\bar{\epsilon}$  & $2\bar{a}-\epsilon$  \\
\hline
$\;\;q=2$  &  $2\bar{a}-\epsilon$   & $2a-\bar{\epsilon}$   \\
\hline
\end{tabularx}},
\end{align}
for any $\tilde{s}\in \tilde{\mathcal{S}}$.

Next, we proceed to calculate the stationary distribution $\pi_{\tilde{S},Q}$, which is necessary for the computation of $I(X,\tilde{S};Y|Q)$. The required transition probability for this computation is given by:
\begin{align}
    P(\tilde{s}^+,q^+|\tilde{s},q)= \sum_{x,y}P(x|\tilde{s},q)P(y|x,\tilde{s})\mathbbm{1}_{\{q^+=\phi(q,y)\}}\mathbbm{1}_{{\{\tilde{s}^+=\tilde{f}(x,y,\tilde{s})\}}}.\nn
\end{align}
Consequently, standard computation of the stationary distribution yields:
\begin{align*}
   \pi_{\tilde{S}|Q}(\tilde{s}|q) = \scalebox{0.93}{\begin{tabularx}{0.3\textwidth} { 
  |@{}p{1.3cm}@{} 
  | >{\centering\arraybackslash}X 
  | >{\centering\arraybackslash}X | }
 \hline
& $\tilde{s}=0$ & $\tilde{s}=1$\\
 \hline
$\;\;q=1$  & $1$  & $0$ \\
\hline
$\;\;q=2$  & $0$  & $1$ \\
\hline
\end{tabularx}}.
\end{align*}

We now verify that the proposed graph-based encoder satisfies the BCJR-invariant property. Let us demonstrate this explicitly for the case where $(q,y)=(1,1)$ and $\tilde{s}^+=1$. Since $\phi(1,1)=2$, the left-hand side of Eq. \eqref{eq: BCJR} is equal to $\pi_{\tilde{S}|Q}(1|2)$, while the right-hand side is equal to:
\begin{align*}
    &\frac{\sum_{x,s}\pi_{\tilde{S}|Q}(s|1)P_{X|\tilde{S},Q}(x|s,1)P_{Y|X,\tilde{S}}(1|x,s)\mathbbm{1}_{\{1=\tilde{f}(x,1,s)\}}}{\sum_{x',s'}\pi_{\tilde{S}|Q}(s'|1)P_{\tilde{X}|\tilde{S},Q}(x'|s',1)P_{Y|X,\tilde{S}}(1|x',s')}\\
    &\stackrel{(a)}= 1,
\end{align*}
where $(a)$ holds because $\mathbbm{1}_{\{1=\tilde{f}(x,1,s)\}}=1$ for any $x,s$, implying the numerator is equal to the denominator. Accordingly, it is indeed equal to $\pi_{\tilde{S}|Q}(1|2)$, as required. The verification of for the other cases can be done in a similar manner.

Finally, the achievable rate of the graph-based encoder is
\begin{align}\label{eq:acheivable_rate}
    R &= I(X,\tilde{S};Y|Q) \nn\\
      &= \sum_{q\in\mathcal{Q}}\pi_Q(q)\cdot I(X,\tilde{S};Y|Q=q) \nn\\
      &= \sum_{q\in\mathcal{Q}}\pi_Q(q)\cdot\left[H_2\left(Y|Q=q\right)-H_2(Y|X,\tilde{S},Q=q)\right]\nn\\
      &\stackrel{(a)}= \sum_{q\in\mathcal{Q}}0.5\cdot\left[H_2\left(a\right)-H_2(Y|X,\tilde{S},Q=q)\right]\nn\\
      &= H_2(a)-\frac{2a-\bar{\epsilon}}{2}H_2\left(\frac{\epsilon}{2}\right)-\frac{2\bar{a}-\epsilon}{2}H_2\left(\frac{1+\epsilon}{2}\right)\nn\\
      &\stackrel{(b)}=\frac{1}{2}\log_2\left(\frac{1}{4}\left(\frac{\epsilon}{1-a}\right)^{\epsilon}\left(\frac{1+\bar{\epsilon}}{a}\right)^{1+\bar{\epsilon}}\right),
\end{align}
where $(a)$ follows due to the fact that
\begin{align}
    P_{Y|Q}(0|q) &= \sum_{x,\tilde{s}} \pi(\tilde{s}|q)P(x|\tilde{s},q)P_{Y|X,\tilde{S}}(0|x,\tilde{s}) \nn\\
    &= \begin{cases} a, & q=1, \\
                  1-a, & q=2,\nn
    \end{cases}
\end{align}
which implies that $\pi_Q(q) = [0.5,0.5]$ due to the symmetry, and $(b)$ follows by simplifying the expression. The proof of Theorem \ref{th: NOST} is concluded by noting that the induced achievable rate in \eqref{eq:acheivable_rate} is equal to the upper bound that we derived at the end of Section \ref{subsec:nost}.
\end{proof}

\section{N-Ising Channel --- Proof of Theorems \ref{Th: N-Ising} and \ref{Th: N-Ising_Q4}} \label{app: N-Ising}
\begin{proof}[Proof of Theorem \ref{Th: N-Ising}]
Consider the first-order Markov $Q$-graph depicted in Fig. \ref{fig:1Markov}. For some $\epsilon\in[0,0.5]$, define the following graph-based test distribution
\begin{align}
    T(y=0|\underline{q})=\left[a,1-a\right],
\end{align}
where
\begin{align*}
    a = \frac{\epsilon^{\epsilon\gamma}}{\epsilon^{\epsilon\gamma}+(\bar{\epsilon}^{\bar{\epsilon}}(1+\bar{\epsilon})^{\epsilon-2})^\gamma (1+\epsilon)^{\gamma(1+\epsilon)}},\;\; \gamma=\frac{1}{2\epsilon^2-3\epsilon+2}.
\end{align*}
Define the constant:
\begin{align}\label{eq: rho_N-Ising}
    \rho^*&=\frac{1}{2+4\bar{\epsilon}}\cdot\log_2\left(\frac{16^\epsilon\epsilon^{\epsilon\bar{\epsilon}}\left(\bar{\epsilon}(1+\bar{\epsilon})\right)^{\epsilon^2-3\epsilon+2}}{64\bar{a}^2(a\bar{a})^{2\bar{\epsilon}}(1+\epsilon)^{\epsilon^2-\epsilon-2}}\right).
\end{align}
Recall that the MDP state is defined as $z_t=(\beta_t,q_t)$, where
\begin{align*}
    \beta_{t}=P(S_t|x_t) = 
    \begin{cases} [1-\epsilon,\epsilon], & x_t=0, \\
                  [\epsilon,1-\epsilon], & x_t=1.
    \end{cases}
\end{align*}
Accordingly, define the following value function:
\begin{align}\label{eq: value_N-Ising}
    h([1-\epsilon,\epsilon],1)&=h([\epsilon,1-\epsilon],2)=0\nn\\
    h([1-\epsilon,\epsilon],2)&=h([\epsilon,1-\epsilon],1)=\frac{1}{2\epsilon-3}\cdot\log_2\left(\frac{\bar{a}\bar{\epsilon}(1+a)}{a(1+\bar{\epsilon})^2}\cdot \left(\frac{a^2(1+\epsilon)(1+\bar{\epsilon})}{\bar{a}^2\epsilon\bar{\epsilon}}\right)^\epsilon\right).
\end{align}

It can be shown that the constant $\rho^*$ in \eqref{eq: rho_N-Ising} and the value function in \eqref{eq: value_N-Ising} solve the Bellman equation under the policy defined below.
\begin{align}\label{eq: policy_N-Ising}
    u^*(\beta,q)=\begin{cases} 1, & q=1, \\
                  0, & q=2.
    \end{cases}
\end{align}
The verification is omitted as it involves the same technical steps performed in proving Theorem \ref{th: NOST}.
\end{proof}
\begin{proof}[Proof of Theorem \ref{Th: N-Ising_Q4}]
Consider a $Q$-graph consists of four nodes with an evolution function that is given by the vectors representation $\underline{\phi}(
\underline{q},y=0)=[2,1,1,1]$ and $\underline{\phi}(\underline{q},y=1)=[3,3,4,3]$. For some $(a,b)\in[0,1]^2$, define the following graph-based test distribution
\begin{align*}
    T(y=0|\underline{q})=\left[a,b,1-a,1-b\right].
\end{align*}
Define the constant:
\begin{align}\label{eq: rho_N-Ising_Q4}
    \rho^*&=\frac{1}{2(1+2\bar{\epsilon})(2+\epsilon)}\cdot\log_2\left(\frac{2^{4\epsilon^2+2\epsilon-12}\bar{a}^{4\epsilon^2-6\epsilon-4}(\epsilon^2-3\epsilon+2)^{\epsilon^3-\epsilon^2-4\epsilon+4}}{a^{4\bar{\epsilon}}b^{2\epsilon\bar{\epsilon}}(1+\epsilon)^{\epsilon^3+\epsilon^2-4\epsilon-4}\epsilon^{\epsilon^3+\epsilon^2-2\epsilon}\bar{b}^{2\epsilon^2-6\epsilon+4}}\right).
\end{align}
Further, define the following value function:
\begin{align}\label{eq: value_N-Ising_Q4}
    h([1-\epsilon,\epsilon],1)&=h([\epsilon,1-\epsilon],3)\nn\\&=\frac{1}{(1+2\bar{\epsilon})(2+\epsilon)}\cdot\log_2\left(\frac{(1+\epsilon)^{\epsilon^2+3\epsilon+2}(2-\epsilon)^{\epsilon^2-4}b^{2\epsilon^3-\epsilon^2-5\epsilon+6}}{a^{2-2\epsilon^2-\epsilon}\bar{b}^{2\epsilon^3-\epsilon^2-5\epsilon+2}\epsilon^{\epsilon^2+2\epsilon}\bar{a}^{2\epsilon^2+\epsilon+2}\bar{\epsilon}^{\epsilon^2+\epsilon-2}}\right)\nn\\
    h([1-\epsilon,\epsilon],2)&=h([\epsilon,1-\epsilon],4)=0\nn\\
    h([1-\epsilon,\epsilon],3)&=h([\epsilon,1-\epsilon],1)=\frac{1}{2+\epsilon}\cdot\log_2\left(\frac{\bar{a}^{\epsilon-2}\bar{b}^{\epsilon^2}}{a^{\epsilon}b^{\epsilon^2-2}}\right)\nn\\
    h([1-\epsilon,\epsilon],4)&=h([\epsilon,1-\epsilon],2)=\log_2\left(\frac{b\bar{b}^\epsilon}{\bar{b}b^\epsilon}\right).
\end{align}
Similarly here, the constant $\rho^*$ in \eqref{eq: rho_N-Ising_Q4} and the value function in \eqref{eq: value_N-Ising_Q4} satisfy the Bellman equation. The verification is omitted and follows the same technical steps.
\end{proof}
\end{appendices}
\bibliography{ref}
\bibliographystyle{IEEEtran}

\end{document}